\documentclass[journal]{IEEEtran}
\usepackage{bbm}
\usepackage{multirow}

\usepackage{graphicx,epsfig}
\usepackage{amsmath,mathrsfs} 
\usepackage{amssymb,amsthm} 
\usepackage{amsfonts}
\usepackage{mathtools}
\usepackage{epstopdf}
\usepackage{ifthen}
\usepackage{bbm}
\graphicspath{{eps/}}
\usepackage{amsthm}
\usepackage{pgfplots}
\usepackage{cite}                                

\newtheorem{Definition}{Definition}
\newtheorem{Corollary}{Corollary}
\newtheorem{Theorem}{Theorem}
 
\newtheorem{Lemma}{Lemma} 
\newtheorem{Proposition}{Proposition}
 
 \UseRawInputEncoding

 \title{The Wavelet Compressibility of \\Compound Poisson Processes}
\author{ Shayan~Aziznejad and~Julien~Fageot 
\thanks{Shayan Aziznejad is with the Biomedical Imaging Group, \'Ecole polytechnique f\'ed\'erale de Lausanne, 1015 Lausanne,
Switzerland (e-mail: shayan.aziznejad@epfl.ch).} 
\thanks{Julien Fageot is with the AudioVisual Communications Laboratory (LCAV), \'Ecole polytechnique f\'ed\'erale de Lausanne, 1015 Lausanne, Switzerland (e-mail: julien.fageot@epfl.ch)}
\thanks{This research was supported by   the Swiss National Science Foundation, Grants 200020\_184646/1, P2ELP2\_181759, and P400P2\_194364.}}
\begin{document}
 \maketitle
\thispagestyle{plain}
\pagestyle{plain}

\begin{abstract}
In this paper, we precisely quantify the wavelet compressibility of compound Poisson processes.  To that end, we  expand the given random process over the Haar wavelet basis and we analyse its asymptotic approximation properties. 
By only considering  the nonzero wavelet coefficients up to a given scale, what we call the  greedy approximation, we exploit the extreme sparsity of the wavelet expansion that derives from the piecewise-constant nature of compound Poisson processes.
More precisely, we provide  lower and upper bounds for the mean squared  error  of greedy approximation of compound Poisson processes. We are then able to deduce that the greedy approximation error has a  sub-exponential and super-polynomial  asymptotic behavior.   Finally, we provide numerical  experiments  to highlight the  remarkable ability of wavelet-based dictionaries  in  achieving highly compressible approximations of  compound Poisson processes. 
\end{abstract}

 
 \section{Introduction}
    
    \subsection{Sparsity and the Limits of Gaussian Models}
    \label{subsec:GaussBad}
    The statistical modelling of data  plays a central role in numerous research domains, such as  signal processing \cite{gray2004introduction} and pattern recognition \cite{webb2003statistical}. 
    In that regard, Gaussian models  have been the first and by far the most considered ones, thanks to their desirable mathematical properties and relatively simple characterization. 
    For instance, the Karhunen-Lo\`eve transform (KLT) identifies the optimal basis for representing data with Gaussian priors~\cite{ahmed2012orthogonal} and Kalman filters  are optimal denoisers of Gaussian signals~\cite{kalman_new_1960}, both in the mean-square sense.
    These facts, among others, have made Gaussian statistical priors very convenient in practice. They also reveal the fundamental relationship between Fourier-based signal representations and Gaussian models. 

    However, it has been a long standing observation that Gaussian models fail to capture several key statistical properties of most naturally-occurring signals~\cite{Srivastava2003advances,Mumford2010pattern}. Indeed, the latter frequently have heavy-tailed  marginals~\cite{Pesquet2002stochastic,Unser2014sparse,Fageot2015statistics,sciacchitano2017image} or richer structure of dependencies than Gaussian ones~\cite{Huang1999statistics,gupta2018generalized}.
     Real-world signals are highly structured  and often admit concise representations, typically on wavelet bases that appear to be genuinely  versatile~\cite{Mallat1999,starck_sparse_2010}. 
    This has led to the current paradigm in modern data science  where \emph{sparsity} plays one of the central roles in  statistical learning~\cite{hastie2015statistical,Starck2013sparse} and signal modelling~\cite{elad2010sparse,rish2014sparse,Unser2014sparse}.
    Classical Gaussian priors cannot  model sparsity as they tend to produce poorly compressible signals~\cite{amini_compressibility_2011,Amini2014sparsity}. Many recent efforts in signal processing have been directed towards  the development of deterministic frameworks that are better tailored to the reconstruction or synthesis of sparse signals, such as traditional compressed sensing~\cite{donoho_compressed_2006,Candes2006sparse,Foucart2013mathematical} and its infinite-dimensional extensions~\cite{adcock2017breaking,Adcock2015generalized,eldar2012compressed,unser2017splines}.

    \subsection{Wavelets and Signal Representations}\label{sec:wavelets} 

The development of wavelet methods, based on the pioneering works of I. Daubechies, Y. Meyer, and S. Mallat in the late 80's~\cite{daubechies1988orthonormal,Mallat1989,meyer1992wavelets}, has shed  new lights on signal representation. Repeated numerical observations confirmed  that wavelet-based compression techniques such as JPEG-2000 \cite{christopoulos2000jpeg2000} outperform classical Fourier-based standards ({\it e.g.}, JPEG) for natural images. This is despite the fact that the discrete Fourier transform (DFT) and its real-valued counterpart, the discrete cosine transform (DCT)  \cite{ahmed_discrete_1974}, are asymptotically equivalent to KLT and, hence, are optimal for representing signals with Gaussian prior \cite{unser1984DCT}. 

Wavelets are celebrated for their excellent approximation properties for large classes of signals and functions~\cite{Devore1998Nterm}. They revived  the field of functional analysis~\cite{meyer1992wavelets,Triebel2008function}, culminating with the Abel prize of Yves Meyer in 2017 and feeding remarkable applications to various scientific and engineering fields~\cite{akansu2010emerging}.
One of the remarkable aspects of wavelets is that they are unconditional bases for many function spaces, including H\"older, Sobolev, and Besov spaces~\cite{meyer1992wavelets,Triebel2008function} which is a key property for studying the best $M$-term approximation in a given basis~\cite{donoho1993unconditional,Devore1998Nterm}. 

    \subsection{Probabilistic Models for Sparse and Analog Signals}
  As we have seen, probabilistic models beyond the Gaussian paradigm are of special interest for the modelling of sparse signals.
    A systematic attempt in this direction has been developed in the monograph~\cite{Unser2014sparse}.
    In this work, analog signals are modelled as solutions of stochastic differential equations driven by non-Gaussian L\'evy white noises. 
     The so-called \it{sparse stochastic processes} have been used to develop novel techniques for essential signal processing tasks, such as denoising \cite{kamilov_mmse_2013} and estimation \cite{bostan2013map} for signals with non-Gaussian priors. These methods have also been used  in biomedical image reconstruction \cite{Bostan.etal2013}, highlighting  the practical aspects of this new statistical framework.
 
    The simplest    class of non-Gaussian model  is the one of compound Poisson processes.
      The latter  are random piecewise constant functions with independent and stationary increments. As such, they are part of the family of L\'evy processes~\cite{bertoin1996levy},   which also includes the  Brownian motion. Compound Poisson processes are fully determined by the heights and locations of their countably many jumps. Contrary to  Brownian motion, they are part of the class of signals with finite rate of innovations~\cite{Vetterli2002FRI,dragotti2007sampling,Unser2011stochastic}, meaning that  their realizations on compact intervals can be fully encoded by finitely many numbers. This makes them particularly appealing for the modelling of highly-compressible piecewise constant signals.
     It has also been shown recently that any L\'evy process is the limit in law of compound Poisson processes whose rate of innovation tends to infinity \cite{fageot_gaussian_2018}. This theoretical observation permits the  development of methods for generating trajectories of L\'evy processes from compound Poisson processes, as exploited in~\cite{dadi2020genssp}.
     
     \subsection{Gaussian versus Poisson: Two Extreme Compressibility Behaviors}
    \label{subsec:GaussVsPoisson}
 
      The aforementioned class of L\'evy processes  (see Section \ref{prelim:levyprocess} for a formal definition) allows for various compressibility behaviors: the Brownian motion is the less compressible, while the compound Poisson ones are at the other extreme. 
        This compressibility hierarchy has been recently revealed in two different theoretical frameworks. 
    
    In the first one, the compressibility is measured via the speed of convergence of the best $M$-term approximation in wavelet bases.
    The decay rate of the best $M$-term error is known to be directly linked to the Besov regularity~\cite{Devore1998Nterm,Garrigos2004sharp}, which has been quantified for a broad class of L\'evy processes \cite{schilling1997feller,schilling1998growth,Veraar2010regularity,fageot2017Besov,fageot2017noisePositive,aziznejad2020wavelet}. Hence, the compressibility of L\'evy processes has already been characterized using this approach \cite{Ward2015SampTA,fageot2017term} and synthesized in \cite[Chapter 6]{fageot2017gaussian}.
    In a nutshell, state-of-the-art results show that the best $M$-term quadratic approximation error of the Brownian motion behaves asymptotically like $1/{M}$\footnote{More precisely, one can deduced from \cite{fageot2017term} that the wavelet approximation error of the Brownian motion decays almost surely faster that $1/M^{1-\epsilon}$ and slower than $1/M^{1+\epsilon}$ for any $\epsilon > 0$ when $M\rightarrow \infty$.}, while the same quantity decays  faster than any polynomial for compound Poisson processes~\cite[Theorems 4 and 5]{fageot2017term}. 

    In the second framework, 
     the compressibility of a L\'evy process is quantified in the information theoretic sense through the entropy of the underlying L\'evy white noise, as in~\cite{ghourchian2018compressible,fageot2020entropic}. These two frameworks are complementary and based on totally different tools, but they are consistent and lead to the same compressibility hierarchy. 
    
     \subsection{Contributions and Outline}
    
    This paper contributes to the analysis of the compressibility of L\'evy processes,  focusing on the compound Poisson    and Gaussian cases. 
    We consider the Haar wavelet approximations of these random processes and quantify the decay rate of their approximation error in the mean squared sense. 
    
    More precisely, we   focus on quantities such as
    \begin{equation} \label{eq:MSEdefinition}
        \mathbb{E} \left[ \lVert s - \mathrm{P}_M \{s\} \rVert_2^2 \right] 
    \end{equation}
    where $s$ is a compound Poisson process or the Brownian motion and $\mathrm{P}_M : L_2([0,1]) \rightarrow L_2([0,1])$ is a possibly nonlinear  approximation operator based on $M \geq 1$ Haar wavelet coefficients of the input function. We compare various approximation schemes, depending on which wavelet coefficients are chosen. The two best-known schemes are the linear and the best $M$-term approximation, albeit both suffer from practical limitations. On one hand, the linear scheme does not capture the sparsity that might be inherent in the signal of interest (see Proposition \ref{Prop:LinError}). On the other hand,  
    in order to {\it exactly} implement a compression scheme based on the best $M$-term approximation of the random process,  one needs to have access to  the full infinite set of
    wavelet coefficients. Without additional  knowledge on the wavelet expansion, the implementation may become cumbersome and not memory efficient, if not impossible.  This is why alternative approximation schemes have been proposed, most notably the ``tree approximation'' scheme which has brought significant attention in the literature \cite{cohen2001tree,baraniuk1999optimal,binev2004fast,baraniuk2002near}. 
    
      In the same spirit, we consider a very simple \it{greedy approximation scheme},
    in which  only  the first $M$  nonzero  wavelet coefficients are preserved. This scheme is well-suited to compound Poisson processes, for which most of the wavelet coefficients are zero due to their piecewise constancy. 

      Our main result is to provide lower and upper-bounds for the greedy approximation error in the mean-squared sense (Theorem \ref{Thm:Main}). 
     It essentially states that the mean-square error  of the Haar greedy approximation of the compound Poisson process $s$ behaves roughly as
    \begin{equation} \label{eq:mouais}
         \mathbb{E} \left[ \lVert s - \mathrm{P}^{\mathrm{greedy}}_M \{s\} \rVert_2^2 \right] \approx  \mathbb{E}  \left[2^{-\frac{ \boldsymbol{M}}{\boldsymbol{N}}}\right]
    \end{equation}
    where $N$ is the (random) number of jumps of $s$ (see \eqref{Eq:MSEbounds} for the precise meaning of \eqref{eq:mouais}).
    This allows us to deduce that the mean-square error decays faster than any polynomial, and slower than any exponential  (Theorem \ref{Thm:SubExpSupPoly}). 
    We also perform a similar analysis for the linear approximation of the compound Poisson process, as well as for  the linear and greedy approximations of the Brownian motion. This    highlights the specificity of the compound Poisson processes: the greedy approximation dramatically outperforms the linear scheme for compound Poisson processes, contrary to the Gaussian case. We summarize this situation in  Table I, where the main contribution is highlighted in bold.

     \begin{table}[t]\label{table:4cases}
	\renewcommand{\arraystretch}{1.3}
	\centering
	\caption{ The decay rate of the MSE \eqref{eq:MSEdefinition}  of  linear, greedy, and best $M$-term approximation schemes for compound Poisson processes and the  Brownian motion.  }
	\begin{scriptsize}
		\begin{tabular}{c|c c}
			\hline
			\hline
			 &  Brownian Motion & Compound Poisson \\
		 	\hline
		 	Linear& $\frac{1}{M}$ & $\frac{1}{M}$ \\ & & \\
            {\bf Greedy} & $\frac{1}{M}$ & $ \pmb{\mathbb{E}}_{\boldsymbol{N}\boldsymbol{\sim} \textbf{Pois}( \boldsymbol{\lambda})}{\bf \left[2^{-\frac{ \boldsymbol{M}}{\boldsymbol{N}}}\right]}$ \\ & & \\
			 Best & $\frac{1}{M}$ & $\ll M^{-k}, \forall k\in \mathbb{N}$ \\
			\hline \hline 
		\end{tabular}
	\end{scriptsize}
\end{table}
 
    Finally, we illustrate our theoretical findings with numerical examples in various cases. Specifically, we highlight that  the approximation error obtained within our method is    close  to the best $M$-term approximation. Moreover, we highlight the role of the wavelet dictionary by comparing the linear and best $M$-term schemes for compound Poisson processes and the Brownian motion in a Fourier-type dictionary corresponding to the discrete cosine transform (DCT).   These empirical observations  raise  interesting theoretical questions which  we briefly expose  and can be exploited in future works.
    
    \subsection{Outline}
    
    The paper is organized as follows: in Section \ref{sec:prelim}, we present the relevant mathematical concepts. We then discuss  our approximation scheme and compare it with the linear and best $M$-term methods in Section \ref{sec:approx}. We  present our main theoretical results in Section \ref{sec:main} and finally, we demonstrate our theoretical results within numerical examples in Section \ref{sec:numerical}. 

\section{Mathematical Preliminaries}\label{sec:prelim}
In this section, we  recall the relevant mathematical concepts and state preparatory results that we will use throughout  the paper.  

\subsection{L\'evy Processes and L\'evy White Noises}\label{prelim:levyprocess}
Brownian motions and compound Poisson processes are members of the general family of L\'evy processes,  which are continuous-domain random processes characterized by their independent and stationary increments \cite{Sato1994levy,bertoin1996levy}.   
L\'evy processes are defined\footnote{This definition,  based on the theory of distributions}, is not the most common one. However, it is proven to be equivalent to more classical ones in \cite{Dalang2015Levy}.  Consequently, we shall view the wavelet coefficients as linear functionals acting on the L\'evy white noises which in effect allows us to simply characterize their probability laws. as the solutions of the stochastic differential equation 
\begin{equation}\label{Eq:Dsw}
\mathrm{D} s = w,
\end{equation}
with the boundary condition $s(0)=0$. In \eqref{Eq:Dsw}, $\mathrm{D}$ denotes the (weak) derivative operator and $w$ is a L\'evy white noise. We choose to only consider zero-mean white noises. However, this comes with no loss of generality as the results are readily extendable to the general case. 

The formal construction of the family of L\'evy white noises as generalized random processes have been exposed in~\cite[Chapter 3]{gelfand_generalized_2014}.
In this framework,   the L\'evy white noises are defined based on their observation through smooth test functions. For each adequate test function $\varphi$, $\langle w , \varphi \rangle$ is then a zero-mean random variable. The collection of random variables $(\langle w, \varphi\rangle)_{\varphi}$ satisfy two important properties:
\begin{itemize}
\item {\bf (Stationarity)} For any test function $\varphi$ and any shift value $t_0\in\mathbb{R}$, the random variables $\langle w, \varphi\rangle$ and $\langle w, \varphi(\cdot-t_0)\rangle$ have the same law. 
\item {\bf (Whiteness)} For any pair of test functions $(\varphi_1,\varphi_2)$ with disjoint support, the random variables  $\langle w, \varphi_1\rangle$ and $\langle w, \varphi_2\rangle$ are independent.
\end{itemize}

The class of valid test functions for a L\'evy white noise have been characterized in~\cite{Rajput1989spectral,fageot2021domain}. It is sufficient for us to know that $\langle w ,\varphi\rangle$ is well defined for any L\'evy white noise $w$ and any square-integrable and compactly supported test function $\varphi$~\cite[Proposition 5.10]{fageot2021domain}. 

The most studied example of L\'evy processes is the Brownian motion, for which $w$ is a Gaussian white noise. In this case, for any $\varphi\in L_2(\mathbb{R})$, the random variable $\langle w,\varphi\rangle$ has a normal distribution with zero mean and variance $\sigma^2 \|\varphi\|_2^2$, where $\sigma^2$ is the variance of the noise \cite[Section 2.5]{gelfand_generalized_2014}. 

Another prominent  subfamily of L\'evy processes are the compound Poisson processes. They are  piecewise constant processes and their statistics is characterized by their probability law of jumps $\mathcal{P}$ and their Poisson parameter $\lambda>0$ that controls the sparsity of the random process (see Figure \ref{fig:CP}).
  \begin{figure}[t]
    \centering
    \includegraphics[ width=0.5\textwidth]{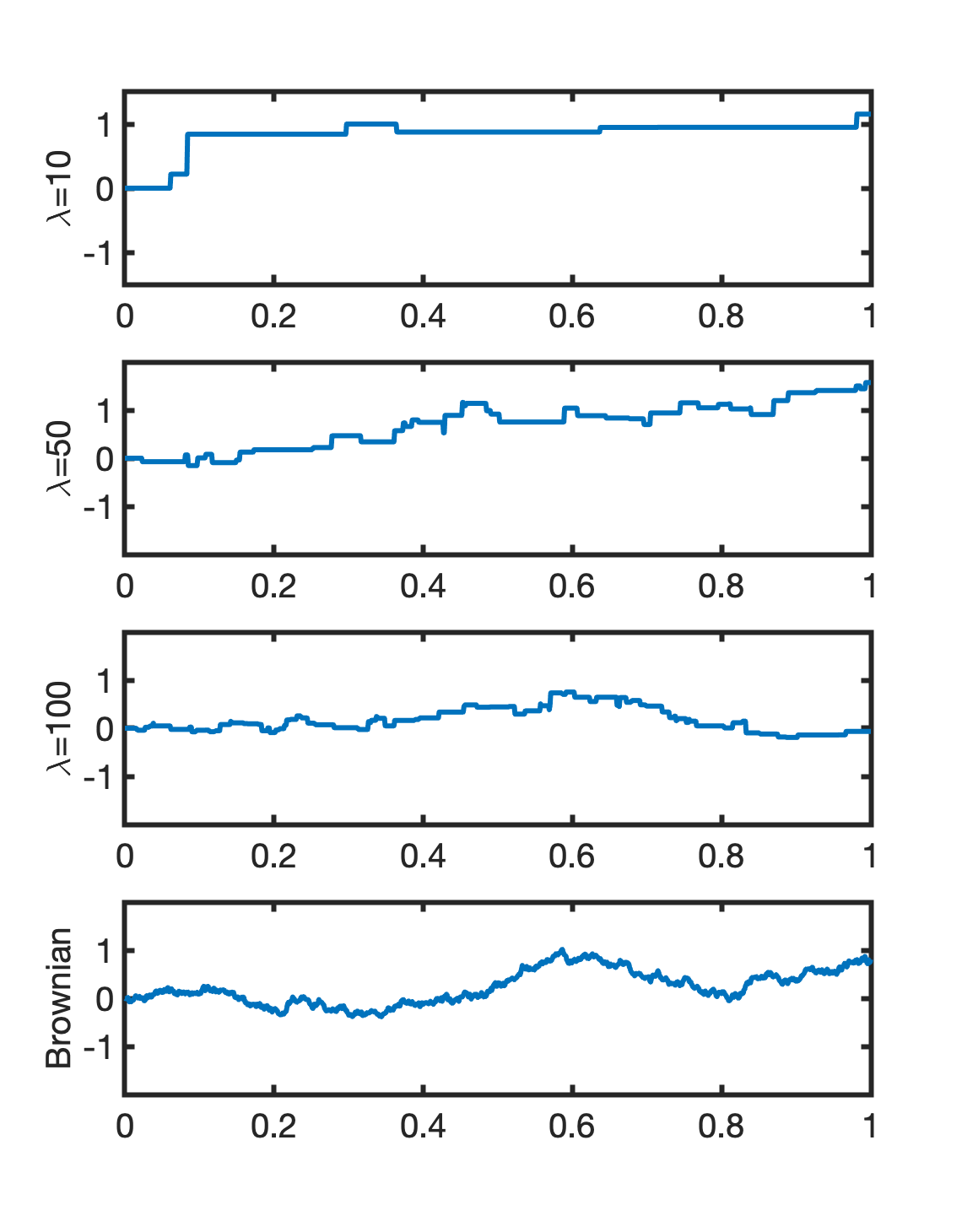}
    \caption{Trajectories of compound Poisson processes with Gaussian jumps  (different values of $\lambda$) and a Brownian motion. All processes are normalized to have unit variance.}
    \label{fig:CP}
\end{figure} 
More precisely, the compound Poisson white noise $w$ with  law of jumps $\mathcal{P}$ and Poisson parameter $\lambda>0$ can be written as a sum of non-uniform Dirac impulses, as 
\begin{equation}\label{Eq:CPNoise}
w= \sum_{k\in\mathbb{Z}} a_k \delta(\cdot-\tau_k),
\end{equation}
where the sequence $\{a_k\}_{k\in\mathbb{Z}}$ of height of Diracs is i.i.d. with  law $\mathcal{P}$ and the sequence $\{\tau_k\}_{k\in\mathbb{Z}}$ of locations of Diracs is a stationary Poisson point-process with parameter $\lambda>0$ (see  \cite{daley2007introduction} for a formal definition of point processes), the $a_k$ and the $\tau_k$ being independent. The key property regarding the Dirac locations is that the number $N$ of $\tau_k$ in any interval $[a,b]$ with $a < b$ is a Poisson random variable with parameter $\lambda (b-a)$. Furthermore, condition to the event $N=n$, the locations  of jumps that are in $[a,b]$ are drawn independently from a  uniform law over $[a,b]$  \cite[Section 2.1]{daley2007introduction}. This implies that,  if we denote by   ${\bf x}=(x_1,\ldots,x_N)$, the ordered set of Dirac locations that are in $[a,b]$, then condition to the event $\{N=n\}$ for any $n\geq 1$, the probability density function (PDF) of ${\bf x}$ is 
\begin{equation}\label{eq:densityt}
p_{\bf x}(\boldsymbol{u}|N=n)= \frac{n!}{(b-a)^n} \mathbbm{1}_{a\leq u_1 \leq\ldots \leq u_n\leq b}
\end{equation} 
for any vector $\boldsymbol{u}=(u_1,\ldots,u_n)\in\mathbb{R}^n$.   It is worth noting that the probability density of $\bf x$, once condition to $N= n \geq 1$, does not depend on $\lambda$ anymore. 
Throughout  the paper, we shall write the ordered jump positions of compound Poisson processes with the letter ${x}$, and the unordered ones with the letter ${\tau}$.

In Lemma \ref{Lemma:Poisson}, we characterize the law of the minimal distance $\Delta$ between two consecutive jumps of a compound Poisson process. The   proof  of Lemma \ref{Lemma:Poisson}  can be found in Appendix \ref{App:LemmaPoisson}.
\begin{Lemma}\label{Lemma:Poisson}
Consider a compound Poisson process $s$ with parameters $(\mathcal{P},\lambda)$ and a fixed  interval $[a,b]$. Denote by $N$, the number of points in $[a,b]$ and  $\mathbf{x}=(x_1,\ldots,x_N)$,  the ordered set of jumps of $s$ that are in $[a,b]$. With the convention  $x_0=a$, we define the random variable $\Delta$ as 
\begin{equation}\label{Eq:DefDelta}
\Delta= \begin{cases} (b-a), & N=0 \\ \underset{1\leq i\leq N}{\min} (x_{i}-x_{i-1}), &N\geq 1. \end{cases}
\end{equation}
Then, almost surely,  $\Delta\leq (b-a)/N$, and  for any $n\geq 1$ and $\delta \in[0,(b-a)/n]$, we have  
\begin{equation}
\mathbb{P}(\Delta \geq \delta|N=n) = \left(1-n\frac{\delta}{b-a}\right)^n.
\end{equation}
\end{Lemma}

 In this paper, we shall consider compound Poisson processes and white noises that are zero-mean, finite variance (which is equivalent to say that the jumps themselves are zero-mean with finite variance), and whose probability law of jumps $\mathcal{P}$ has a PDF (in particular, it has no atoms, what will be used in our analysis). The prototypical example is a compound Poisson process with Gaussian jump heights. 
 
 Despite the fact that their sample paths have very distinct behaviors (see Figure \ref{fig:CP}), finite-variance compound Poisson processes have the same second-order statistics as the Brownian motion. Indeed, for any test function $\varphi$, the random variable $\langle w , \varphi \rangle$ has zero-mean and variance $\sigma^2 \lVert \varphi \rVert_2^2$ for any L\'evy white noise with   variance $\sigma^2$ and zero mean \cite[Proposition 4.15]{Unser2014sparse}.

\subsection{Haar Wavelets} 
 For a pair of functions $\psi,\phi \in L_2(\mathbb{R})$, that are referred to as the mother and father  wavelets, respectively,  the wavelet family $\boldsymbol{{\Psi}}$ contains all (normalized) dyadic scales and integer shifts of $\psi$ plus the integer shifts of $\phi$.  In other words,  we have that $\boldsymbol{{\Psi}}=\{\psi_{j,k}\}_{j\geq 0, k\in\mathbb{Z}} \cup \{\phi_k\}_{k\in\mathbb{Z}}$, where
\begin{align}
&\psi_{j,k}= 2^{\frac{j}{2}} \psi(2^j \cdot- k),
&\phi_k=\phi(\cdot-k),
\end{align}
for all scaling factor $j\geq 0$ and all shifting parameter $k\in\mathbb{Z}$. 

We consider the family of Haar wavelets 
whose mother and father wavelets are respectively
\begin{align}\label{Eq:HaarWavelet}
\psi= \mathbbm{1}_{[0,\frac{1}{2}]}-\mathbbm{1}_{[\frac{1}{2},1]} \quad \text{and} \quad  \phi= \mathbbm{1}_{[0,1]}.
\end{align}
Haar wavelets are known to form an orthonormal basis for $L_2(\mathbb{R})$ \cite{daubechies1988orthonormal}. This means that any function $f\in L_2(\mathbb{R})$   admits the unique expansion 
\begin{equation}\label{Eq:WaveletExpansion}
f(\cdot)=\sum_{j\geq 0} \sum_{k\in\mathbb{Z}}  \langle f,\psi_{j,k}\rangle \psi_{j,k}(\cdot) + \sum_{k\in\mathbb{Z}} \langle f, \phi_k\rangle \phi_k(\cdot),
\end{equation}
where $\langle \cdot,\cdot \rangle$ denotes the standard inner product in $L_2(\mathbb{R})$, defined as $\langle \varphi_1,\varphi_2 \rangle = \int_{\mathbb{R}} \varphi_1(x) \varphi_2(x) \mathrm{d}x$. 

  The  simple characteristics and  implementation of Haar wavelets make them favorable in practice \cite{porwik2004haar,stankovic2003haar}.  They are also compactly supported, which is of great importance in our analysis, due to the {\it whiteness} property of L\'evy white noises (see above). Last but not least, the family consists of piecewise constant functions. Hence, it is  natural  to represent compound Poisson processes (that are themselves almost surely piecewise constant) in this basis. 
 
    \subsection{Haar Decomposition of L\'evy Processes}

In the sequel, we restrict both the random processes and the wavelet transforms to $[0,1]$ and study  the local compressibility  of compound Poisson processes over this compact interval.

Due to the support localization of the Haar wavelets, we readily see that the family $\boldsymbol{{\Psi}}=\{\psi_{j,k}\}_{j\geq 0, 0 \leq k \leq 2^j-1} \cup \{\phi\}$ forms an orthonormal basis of $L_2([0,1])$, hence the L\'evy process $s$ can be almost surely written as
\begin{equation}\label{Eq:WaveletExpansionPoisson}
s= \langle s, \phi\rangle \phi + \sum_{j\geq 0} \sum_{k=0}^{2^j-1} \langle s,\psi_{j,k}\rangle \psi_{j,k}.
\end{equation}
The  probability  law of the Haar wavelet coefficients of $s$ has been characterized in~\cite{Fageot2015statistics}, where their characteristic functions have been explicitly computed. Here, we study the law of wavelet coefficients using the properties of  the underlying L\'evy white noise. In order to achieve this goal, we introduce the auxiliary functions  defined for $t \in [0,1]$ as
 \begin{align}\label{Eq:phitield}
    \tilde{\phi} (t) &= (1 - t) \mathbbm{1}_{[0,1]} (t), \quad \text{and}  \\
 	\tilde{\psi}_{j,k} (t) &=  \begin{cases} 2^{j/2} ( k/2^j-t ), &   t \in [\frac{k}{2^j}, \frac{k+1/2}{2^j}) \\
						2^{j / 2} (t- (k+1)/2^j), & t \in [\frac{k+1/2}{2^j}, \frac{k+1}{2^j}) \\
						0, & \mbox{otherwise,}
					\end{cases} \label{Eq:psitield}	
\end{align}
for any $j\geq 0$ and $k=0,\ldots,2^j-1$. We conclude this part with Proposition \ref{prop:HaarTield}, that expresses the Haar wavelet coefficients of $s$ using the underlying L\'evy white noise and the auxiliary functions  \eqref{Eq:phitield} and \eqref{Eq:psitield}. The proof is available in Appendix \ref{App:HaarTield}.
\begin{Proposition}\label{prop:HaarTield}
Let $s$ be a L\'evy process. Then, for any $j\geq 0$ and $0 \leq k \leq 2^j -1$, we have
\begin{equation}\label{Eq:InnovWavelet}
\langle s, \psi_{j,k} \rangle = \langle w, \tilde{\psi}_{j,k}\rangle, \qquad \langle s, \phi \rangle = \langle w, \tilde{\phi}\rangle,
\end{equation}
where $w$ is the L\'evy white noise such that $\mathrm{D}s=w$.
\end{Proposition}

\section{Wavelet-Based Approximation Schemes}\label{sec:approx}

In this section, we consider three different approximation schemes for square-integrable functions over $[0,1]$:  the linear, best $M$-term, and greedy approximations. Our main goal and contribution is to precisely quantify the approximation power of the greedy scheme.

    \subsection{Wavelet-Based Approximation Schemes}
    
In what follows, we consider the natural indexing of  wavelets by defining the indexing function ${\rm Ind}: \boldsymbol{\Psi}\rightarrow \mathbb{N}$ as 
\begin{equation} \label{eq:ind}
{\rm Ind}(\phi) = 0, \quad {\rm Ind}(\psi_{j,k}) = 2^j + k,  
\end{equation}
for all $j\geq 0$ and $k=0,\ldots,2^j-1$. 

 \begin{Definition}\label{Def:ApproximationSchemes}
Let $f \in L_2([0,1])$. We denote by
\begin{itemize}
\item $\mathrm{P}_{M}^{\rm lin}(f)$, the  {\bf linear} approximation of $f$, that is obtained by keeping the {\bf first} $M$ wavelet coefficients (with respect to the indexing function ${\rm Ind}$) of $f$ in the expansion \eqref{Eq:WaveletExpansionPoisson}. 
\item  $\mathrm{P}_M^{\rm best}(f)$, the {\bf best} $M$-term approximation of $f$, that is obtained by keeping the $M$ {\bf largest} wavelet coefficients of $f$.
\end{itemize}
\end{Definition} 

 The first scheme in Definition \ref{Def:ApproximationSchemes} is called linear due to the fact that $\mathrm{P}_M^{\rm lin}(f)$ depends linearly on $f$. However, the best $M$-term approximation is adaptive to the signal and  is  therefore nonlinear. One can hope that the adaptiveness of the best $M$-term approximation significantly improve the quality of the approximation when compared with the linear one, what appears to be the case for some classes of functions~\cite{Devore1998Nterm}.
 
    As an alternative approach, we consider a compression scheme for compound Poisson processes that can be performed in an online fashion with respect to the stream of the wavelet coefficients. The main idea is to exploit the tremendous sparsity of the expansion of compound Poisson processes over the Haar wavelet basis, that is done by retaining only the nonzero wavelet coefficients and is called the \emph{greedy approximation}.

\begin{Definition}\label{Def:ApproximationSchemesparse}
Let $f \in L_2([0,1])$. We denote by
  $\mathrm{P}_M^{\rm greedy}(f)$, the  {\bf greedy} approximation of $f$, where only the $M$ {\bf first  nonzero} wavelet coefficients are preserved (the ordering being understood with respect to the indexing function ${\rm Ind}$ in \eqref{eq:ind}).
\end{Definition}

As for the best $M$-term approximation, the greedy approximation of $f$ is nonlinear with respect to $f$.   However, it is greedy in the sense that it can be computed by simply looking at the ordered wavelets coefficients. Hence, it does not necessitate to observe the complete set of wavelet coefficients, contrary to the best $M$-term approximation. It therefore shares the simplicity of the linear scheme and the adaptiveness of the optimal scheme (the best $M$-term).

The three approximation schemes introduced in this section   clearly satisfy the relations
\begin{equation} \label{eq:bounds}
    \lVert f - \mathrm{P}_M^{\rm best}(f) \rVert_2 
    \leq 
    \lVert f - \mathrm{P}_M^{\rm greedy}(f) \rVert_2 
    \leq
    \lVert f - \mathrm{P}_M^{\rm lin}(f) \rVert_2 
\end{equation}
for any function $f \in L_2([0,1])$ and any $M \geq 0$.

    \subsection{Mean-Squared Error of the Wavelet Approximations}

 Let $s$ be a L\'evy process.
To quantify the performance of an approximation scheme, we consider the mean-squared error (MSE), which we denote by ${\rm MSE}_{M}^{\rm method}$ for the approximation scheme ${\rm method}\in \{ \rm lin,greedy,best\}$ and is   defined as 
\begin{equation} \label{eq:MSEmethod}
{\rm MSE}_{M}^{\mathrm{method}} = \mathbb{E}\left[\|s-\mathrm{P}_M^{\rm method}(s) \|_{L_2}^2\right].
\end{equation}
It is clear from \eqref{eq:bounds}  that 
\begin{equation}\label{Eq:BleNleL}
{\rm MSE}_{M}^{\mathrm{best}} \leq {\rm MSE}_{M}^{\mathrm{greedy}}  \leq {\rm MSE}_{M}^{\mathrm{lin}}. 
\end{equation}

    \subsection{The Linear Scheme}

In Proposition \ref{Prop:LinError}, we determine the ${\rm MSE}_{M}^{\mathrm{lin}}$ of any L\'evy process that has  finite variance. Its proof is available in Appendix \ref{App:linerr}.
\begin{Proposition}\label{Prop:LinError}
Let $s$ be a L\'evy process with finite variance $\sigma_0^2$. Then, for every $M \geq 1$, we have 
\begin{equation}\label{Eq:linerrbis}
{\rm MSE}_{M}^{\mathrm{lin}}= \frac{\sigma_0^2}{12} \frac{1}{2^J} \left(2 - \frac{m}{2^J}\right),
\end{equation}
where $J=\lfloor\log_2 M\rfloor$ and $m = M-2^J \in \{0,\ldots,2^J-1\}$. In particular, for every $M \in 2^\mathbb{N}$, we have that
\begin{equation}\label{Eq:linerr}
{\rm MSE}_{M}^{\mathrm{lin}}= \frac{\sigma_0^2}{6M}.
\end{equation}
\end{Proposition}

Proposition \ref{Prop:LinError} shows that the linear approximations of L\'evy processes with finite variance   share the same mean-square error.
Let us also remark that if  $s$ is a Brownian motion, then the random variables  $X_{j,k} = \langle s,\psi_{j,k} \rangle=\langle w,\tilde{\psi}_{j,k} \rangle $   are all Gaussian.
Hence, 
$$
\mathbb{P}(\exists j,k: X_{j,k}=0) \leq \sum_{j\geq 0} \sum_{k=0}^{2^j-1} \mathbb{P}(X_{j,k}=0) =0, 
$$
and all the countably many wavelet coefficients are almost surely nonzero and hence, the linear and greedy schemes coincide, as stated in Corollary \ref{corol:LinSparse}. 
  \begin{Corollary}\label{corol:LinSparse}
Let $s$  be a Brownian motion.  Then, for any  $M \geq 0$, we have  the almost sure relation 
\begin{equation}
 {\rm MSE}_{M}^{\mathrm{greedy}} ={\rm MSE}_{M}^{\mathrm{lin}}.
\end{equation}
\end{Corollary}

    \subsection{The Greedy Approximation of Compound Poisson Processes}   

When the wavelet coefficients are sparse ({\it i.e.} when at each scale, only a few of them are nonzero),  the linear and greedy approximation schemes are no longer identical. 
In Proposition \ref{Proposition:Wavelet}, we study the sparsity of the wavelet coefficients of compound Poisson processes.  Precisely, we first characterize when a specific wavelet coefficient vanishes, depending on the presence of jumps. Using this primary result,  we provide upper and lower bounds for the minimal (random) scale at which at least $M$ wavelet coefficients are nonzero. The proof of Proposition \ref{Proposition:Wavelet} is provided in Appendix \ref{App:PropositionWavelet}.  

\begin{Proposition}\label{Proposition:Wavelet}
 Let $s$ be a compound Poisson process  
      whose law of jumps admits a PDF with zero-mean and finite variance. 
\begin{enumerate}
\item For all $j\geq 0$ and $k=0,\ldots,2^j-1$, denote $K_{j,k}$ as the random number of jumps of $s$ in the support of $\psi_{j,k}$. Then,  we almost surely have
\begin{equation}\label{Eq:14}
\langle s, \psi_{j,k} \rangle = 0 \quad \Leftrightarrow \quad K_{j,k} =0. 
\end{equation}
In other words, the symmetric difference between  the events $\langle s, \psi_{j,k} \rangle = 0$ and $K_{j,k} =0$ has probability zero. 
\item Consider the wavelet expansion \eqref{Eq:WaveletExpansionPoisson} of $s$ and denote by $N_J$,  the random number of nonzero wavelet coefficients with scale no larger than $J$.  Furthermore, condition to $\{N\geq 1\}$,  let $J_M$ be the smallest random value of $J$ such that $N_{J}\geq M$; that is, $J_M$ is characterized by $N_{J_{M}-1}<M\leq N_{J_M}$. Then, we have
\begin{align}\label{Eq:JmBound}
\left\lceil \frac{M-2}{N}\right\rceil \leq J_M \leq \left\lfloor \frac{M-1}{N}   +\log \Delta^{-1}\right\rfloor,
\end{align}
where the random variable $\Delta$ is defined in \eqref{Eq:DefDelta}.
\end{enumerate}
\end{Proposition}
    
\section{Compressibility of Compound Poisson Processes}\label{sec:main}
In this section, we present our main result on characterizing the asymptotic behavior of the  greedy approximation of compound Poisson processes. 

\begin{Theorem}\label{Thm:Main}
Let $s$ be a compound Poisson process with  Poisson parameter $\lambda>0$ whose law of jumps  admits a PDF with zero-mean and finite variance.   
 Then for every $M\in \mathbb{N}$,  we have that
\begin{equation}\label{Eq:MSEbounds}
 C_1 M^{-1} \mathbb{E}[2^{-\frac{M}{N}}] \leq {\rm MSE}_{M}^{\rm greedy} \leq  C_2 M  \mathbb{E}[2^{-\frac{M}{N}}],
\end{equation}
where $N$ is a Poisson random variable with parameter $\lambda$, and  $C_1,C_2>0$ are some constants. 
\end{Theorem}
The proof can be found in Appendix \ref{App:Main}. Here, we  give a sketch of the proof. 
For an arbitrary fixed integer $n\geq 1$, we work conditionally to $N=n$. From  the definition of $J_M$ (see Proposition \ref{Proposition:Wavelet}), one has that  $N_{J_M-1} \leq M-1$. Hence the $M$th nonzero wavelet coefficient is reached at scale $J_M$, and therefore 
\begin{equation}\label{Eq:JmSpLinBound}
  \|s- \mathrm{P}^{\rm lin}_{2^{J_M+1}}(s) \|_2\leq \|s- \mathrm{P}_{M}^{\rm greedy}(s) \|_2   \leq \|s- \mathrm{P}^{\rm lin}_{2^{J_M}}(s) \|_2, 
\end{equation}
almost surely. From Proposition \ref{Prop:LinError}, we know the exact behavior of the linear approximation error. On the other hand, we have lower and upper-bounds for the quantity $J_M$, thanks to Proposition \ref{Proposition:Wavelet}. The rest of the proof   leverages  these two preliminary results in order to derive the announced bounds. 

Theorem \ref{Thm:Main} provides lower and upper bounds for the greedy approximation error of any finite-variance compound Poisson process.  In Theorem \ref{Thm:SubExpSupPoly}, we use these bounds to deduce sub-exponential super-polynomial behaviors for the greedy approximation error of compound-Poisson processes. 

\begin{Theorem}\label{Thm:SubExpSupPoly}
Let $s$ be a compound Poisson process  
  whose law of jumps admits a PDF  with zero-mean and finite variance. Then the greedy approximation error ${\rm MSE}^{\rm greedy}_{M}$ of $s$ follows a {\it sub-exponential} and {\it super-polynomial} asymptotic behavior. Precisely, for any $k\in \mathbb{N}$, we have that 
\begin{equation}\label{Eq:SuperPoly}
\lim_{M\rightarrow +\infty } M^k {\rm MSE}_M^{\rm greedy} =0,
\end{equation}
and for any $\alpha>0$, 
\begin{equation}\label{Eq:SubExpo}
\lim_{M\rightarrow +\infty } e^{\alpha M} {\rm MSE}_M^{\rm greedy} =+\infty.
\end{equation}
\end{Theorem}
The proof of Theorem \ref{Thm:SubExpSupPoly} can be found in Appendix \ref{App:SubExpSupPoly}. 
 An enlightening  consequence of the super-polynomial behavior of the greedy approximation error is that it demonstrates that our provided lower- and upper-bounds are asymptotically comparable. Specifically from the upper-bound provided in Theorem \ref{Thm:Main}, we deduce that 
$$\frac{{\rm MSE}_{M}^{\rm greedy}}{\mathbb{E}[2^{-\frac{M}{N}}]^{(1-\epsilon)}} \leq C_2 M{\mathbb{E}[2^{-\frac{M}{N}}]^{\epsilon}} $$
for any $\epsilon>0$. Moreover, Theorem \ref{Thm:SubExpSupPoly} implies that the quantity $M{\mathbb{E}[2^{-\frac{M}{N}}]^{\epsilon}}$ tends to 0 as $M\rightarrow +\infty$ and is therefore bounded from above. Using a similar argumentation for the lower-bound of Theorem \ref{Thm:Main}, we obtain the following corollary.
\begin{Corollary}\label{Prop:Tight}
  For any $\epsilon>0$, there are positive constants $C_{1,\epsilon},C_{2,\epsilon}>0$ such that 
\begin{equation}\label{Eq:Tightness}
     C_{1,\epsilon} \mathbb{E}[2^{-\frac{M}{N}}]^{(1+\epsilon)} \leq {\rm MSE}_{M}^{\rm greedy} \leq  C_{2,\epsilon} \mathbb{E}[2^{-\frac{M}{N}}]^{(1-\epsilon)},
\end{equation}
for all values of $M\geq 1$. 
\end{Corollary}
 
  Our theoretical analysis validates  the two following observations in a rigorous manner: 
\begin{itemize}
    \item A piecewise constant function with a fixed number of jumps $n \geq 1$ is such that its greedy approximation in the Haar basis roughly behaves like $\mathcal{O}(2^{- M /n})$, which is exponential and therefore decays to 0 faster than any polynomial. Note that the exponential decay is faster for smaller values of $n$.
    \item The number of jumps $N$ of a compound Poisson is random. It is almost surely finite but can be arbitrarily large. The concrete effect is that the mean-square error of the greedy approximation roughly behaves like $\mathcal{O}(\mathbb{E}[2^{- M /N}])$.   The subexponential behavior of the MSE is then a consequence, as we have shown. 
\end{itemize}

It is worth noting that the characterization provided by Theorem \ref{Thm:SubExpSupPoly} is not deducible from earlier works that was based on the machinery of Besov regularity, such as \cite{fageot2017term}.   Previous works focus on the almost sure behavior of the approximation error, while we focus on the mean-square approximation on this paper. These are two different regimes and  to the best of our knowledge, the possible link between the two has not been investigated. 

By contrast, we obtain some information regarding the asymptotic behavior of best M-term approximation error of compound Poisson processes from Theorem \ref{Thm:SubExpSupPoly}. Indeed,  by combining \eqref{Eq:BleNleL} and \eqref{Eq:SuperPoly}, one  observes that 
\begin{align*}
 {\lim\sup}_{M\rightarrow +\infty } M^k {\rm MSE}_M^{\rm best} \leq  \lim_{M\rightarrow +\infty } M^k {\rm MSE}_M^{\rm greedy}=0,
\end{align*}
for any $k\in\mathbb{N}$.  Using the fact that MSE is non-negative simply implies that 
\begin{equation}
\lim_{M\rightarrow +\infty } M^k {\rm MSE}_M^{\rm best} = 0.
\end{equation}
  The super-polynomial decay of the greedy approximation error shows that this method, despite being very simple and easily implemented, reaches excellent approximation performances. In the next section, we will empirically show that the greedy scheme performs similarly to the practically uncomputable best $M$-term approximation scheme. 

\section{Numerical Illustrations}\label{sec:numerical}

In this section, we provide a numerical demonstration of  the main results of this paper. 
First, it is illustrative and reflects the potential practical impact of our theoretical claims in a complementary and empirical manner.
Second, it shows that the results obtained for the greedy approximation method are similar to what would be obtained for the best $M$-term approximation (Section \ref{sec:sparsevsbest}).
Finally, it   emphasizes that wavelets are able to exploit the inherent sparsity of non-Gaussian signals, which is not the case of traditional Fourier-based approximation schemes (Section \ref{sec:HaarFourier}).  These empirical observations then give rise to open theoretical questions that might be of interest to the community. 

To simulate each approximation 
scheme, we first generate a 
signal that consists of   $2^{10} = 1024$ equispaced samples of a 
given random process over 
$[0,1]$. We then compute its 
(discrete) Haar wavelet 
coefficients of scale up to 
$J_{\max}=10$. Finally, we create
the approximated signal according
to the given approximation scheme\footnote{For the best $M$-term approximation, we do not have access to the infinitely many wavelet coefficients but only to the ones up to a given scale ($J_{\max} = 10$ in this case). This means that we only have an approximation of the best $M$-term for our simulations. However, the variance of the wavelet coefficients decay with the scale $j$ like $2^{-2j}$ and the coefficients at larger scales are therefore very small with high probability. Our approximation of the best $M$ terms  is therefore  excellent.}. We repeat  each experiment 1000 times and we report the average to reduce the effect of the underlying randomness  (Monte Carlo method). The averaged values are then good approximations of the quantities of interest, that is, the MSEs given by \eqref{eq:MSEmethod} for   different approximation schemes.

\subsection{Greedy Approximation Error}
In the first experiment, we compute the MSE of greedy approximation for Brownian motions and compound Poisson processes with different values of $\lambda=10,50,100,500$ and with Gaussian jumps, as a function of the number $M$ of coefficients that are preserved.  We recall that, $N$ being the random number of jumps of the compound Poisson process $s$ over $[0,1]$,
$\lambda = \mathbb{E} [N]$ is the averaged number of jumps.
To have a fair comparison, we unify the variance of the random processes in all cases to be $\sigma_0^2=1$  (which corresponds to a law of jumps with variance $\sigma_0^2 / \lambda = 1 / \lambda$ for compound Poisson processes). 

The results are depicted in Figure \ref{fig:Nonlin}, where in each case  we plot the MSE   in log scale, that is $\log_2(M) \mapsto 10\log_{10}({\rm MSE})$.   From  Proposition \ref{Prop:LinError} and Corollary \ref{corol:LinSparse}, we expect that  the MSE of Brownian motion follows a global linear decay in the log scale, while decaying sub-linearly locally.
Indeed, for $M=2^J$, $J\in\mathbb{N}$, we deduce from \eqref{Eq:linerr} that 
$$10\log_{10}({\rm MSE}_M^{\rm greedy})= \alpha- \beta J, $$
where $\alpha=10 \log_{10}(\sigma_0^2/6)$ and $\beta=10 \log_{10}(2)$ which shows a linear decay with respect to $J=\log_2(M)$. However, in the regime when $J=\lfloor \log_2(M)\rfloor$ is fixed, that is when $2^J \leq M<2^{J+1}$, we obtain from   \eqref{Eq:linerrbis} that 
$$10\log_{10}({\rm MSE}_M^{\rm greedy})= \alpha- \beta (J+1)  
+ \beta \log_2\left(3 - \frac{M}{2^J}\right),$$ 
 which shows that the error decays  sub-linearly in this regime. These theoretical claims can be observed  in Figure \ref{fig:Nonlin}, as well. 

In addition, from Theorem \ref{Thm:SubExpSupPoly}, we  know that the MSE of compound Poisson processes in the log scale should   asymptotically decay faster than any straight line. This is also observable in Figure \ref{fig:Nonlin}, indicating  the dramatic difference between the compressiblity of compound Poisson processes and Brownian motions, as expected. 

\begin{figure}[t]
    \centering
    \includegraphics[ width=0.5\textwidth]{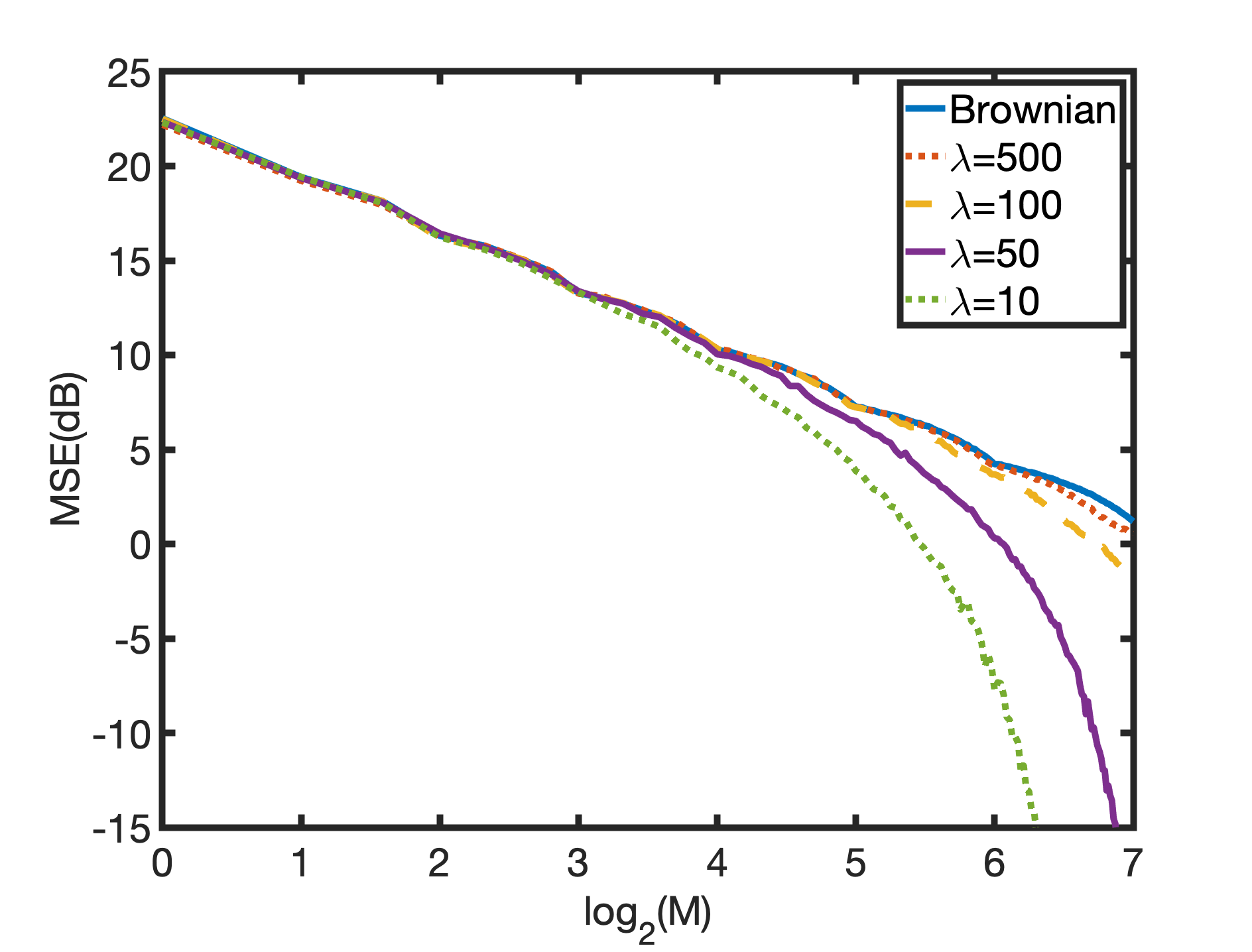}
    \caption{  Greedy approximations of Brownian motions and compound Poisson processes with different values of $\lambda$ and Gaussian jumps. We fix the variance to one in all cases. }
    \label{fig:Nonlin}
\end{figure}
We moreover remark in Figure \ref{fig:Nonlin} that the small-scale behavior ($\log_2(M) = J \leq 3$) does not distinguish between different values of $\lambda$, but also between compound Poisson processes and the Brownian motion. Again, this empirical fact has a theoretical counterpart: it is linked with the fact that the statistics of finite variance compound Poisson processes are barely distinguishable from the ones of the Brownian motion at coarse scales. This has been formalized in~\cite{fageot2019scaling} which states, when particularized to our case, that compound Poisson processes with finite variance converge to the Brownian motion when zoomed out and correctly renormalized. Our numerical experiments are illustrative to this point, and will be confirmed in Sections \ref{sec:sparsevsbest} and \ref{sec:HaarFourier}.

Finally, we observe  in Figure \ref{fig:Nonlin} that as $\lambda\rightarrow +\infty$, the greedy approximation of compound Poisson processes converges pointwise to the one of Brownian motion.  This empirical observation poses an interesting theoretical question which is also    consistent with~\cite[Theorem 5]{fageot_gaussian_2018}, which states---when specialized to our problem---that the compound Poisson process with constant variance $\sigma_0^2$ and Gaussian jumps converges in law to the Brownian motion when $\lambda \rightarrow \infty$.

\subsection{Greedy vs. Best $M$-term Approximation}
\label{sec:sparsevsbest}

As we have seen in the introduction, it is particularly satisfactory to characterize the compressibility of L\'evy processes via their best $M$-term approximation error in a given basis. Although our greedy approximation error only provides an upper-bound for the best $M$-term approximation error, we demonstrate numerically in Figure \ref{fig:LinNonlin} that the two approximation schemes are comparable in the sense of MSE.  This is also an important observation, as it reveals that the extremely simple greedy approximation performs  almost  as good as the best $M$-term approximation, the latter being a theoretical bound for M-term approximation schemes.

   \begin{figure}[t]
    \centering
    \includegraphics[ width=0.5\textwidth]{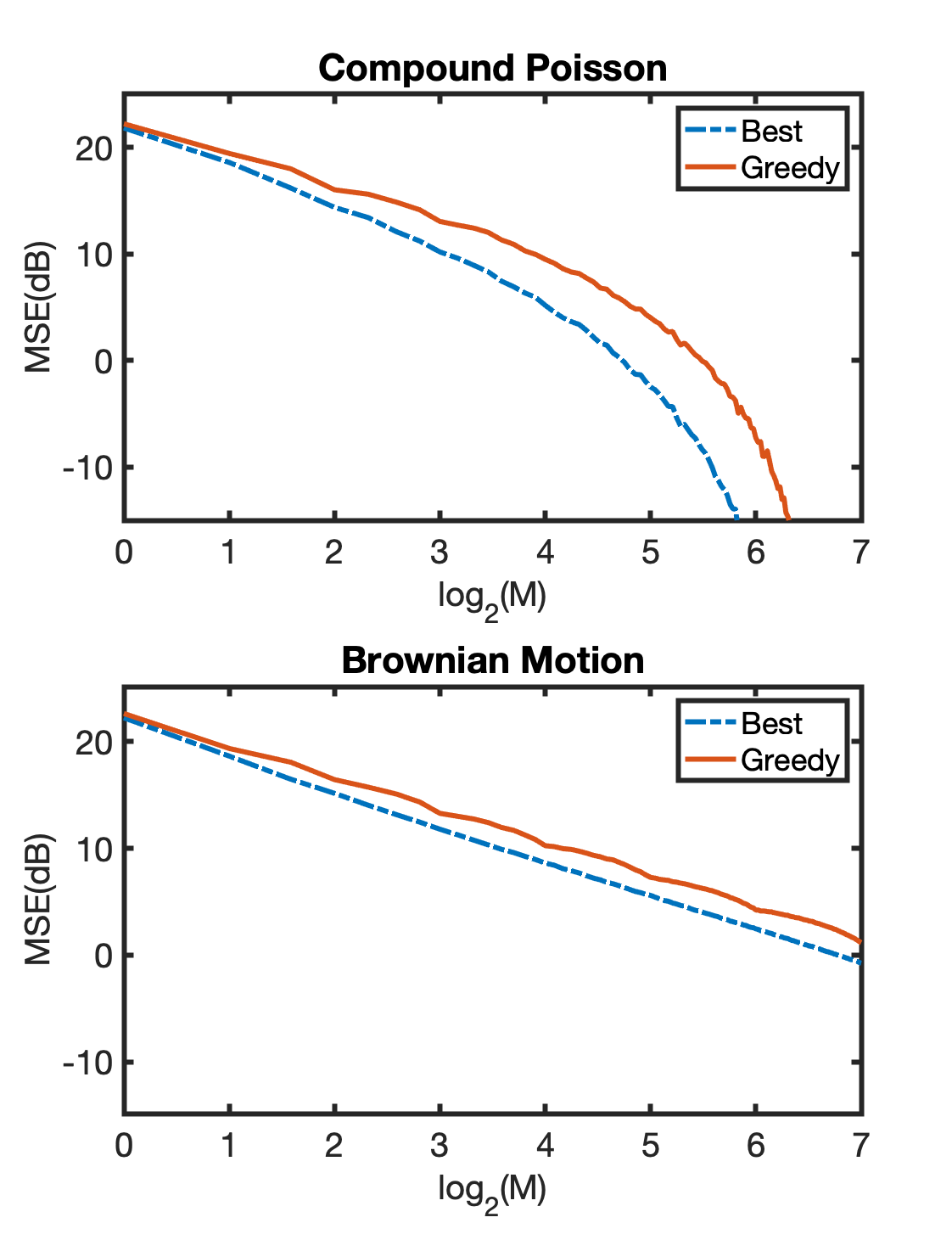}
    \caption{Greedy and Best $M$-term approximation of a compound Poisson process (top) with $\lambda=10$ and Gaussian jumps with a Brownian motion (bottom). We normalize both processes to have unit variance}
    \label{fig:LinNonlin}
\end{figure}
\subsection{Haar vs. Fourier}
\label{sec:HaarFourier}
We now investigate the effect of the dictionary in which we perform the approximation scheme. We consider  the Haar transform and discrete cosine transform (DCT) for approximating the Brownian motion and compound Poisson processes with Gaussian jumps.  The results are depicted in Figure \ref{fig:DCTvsHaar}, where we plot the best $M$-term approximation error of each setup in the log scale.

We observe that the DCT works slightly better than Haar for the Brownian motion. 
This is not surprising: The DCT is known to be asymptotically equivalent to the Karhunen-Lo\`eve transform (KLT), which is optimal for \emph{Gaussian} stationary processes~\cite{unser1984DCT}. It is worth noting that this is also valid for the Brownian motion, which is not stationary but still admits stationary increments.

 However, there is a dramatic difference between Haar and DCT for compound Poisson processes.
We see in Figure~\ref{fig:DCTvsHaar} that, contrary to the Haar dictionary, the DCT is unable to take advantage of the effective sparsity of compound Poisson processes. This is of course not a surprise and is folklore knowledge, but it has not yet been justified theoretically  for the best of our knowledge. 
 This is nevertheless  consistent with recent theoretical and empirical results demonstrating that wavelet methods outperform classical Fourier-based methods for the analysis of sparse stochastic processes~\cite{Unser2014sparse,fageot2017term}. 
   \begin{figure}[t]
    \centering
    \includegraphics[ width=0.5\textwidth]{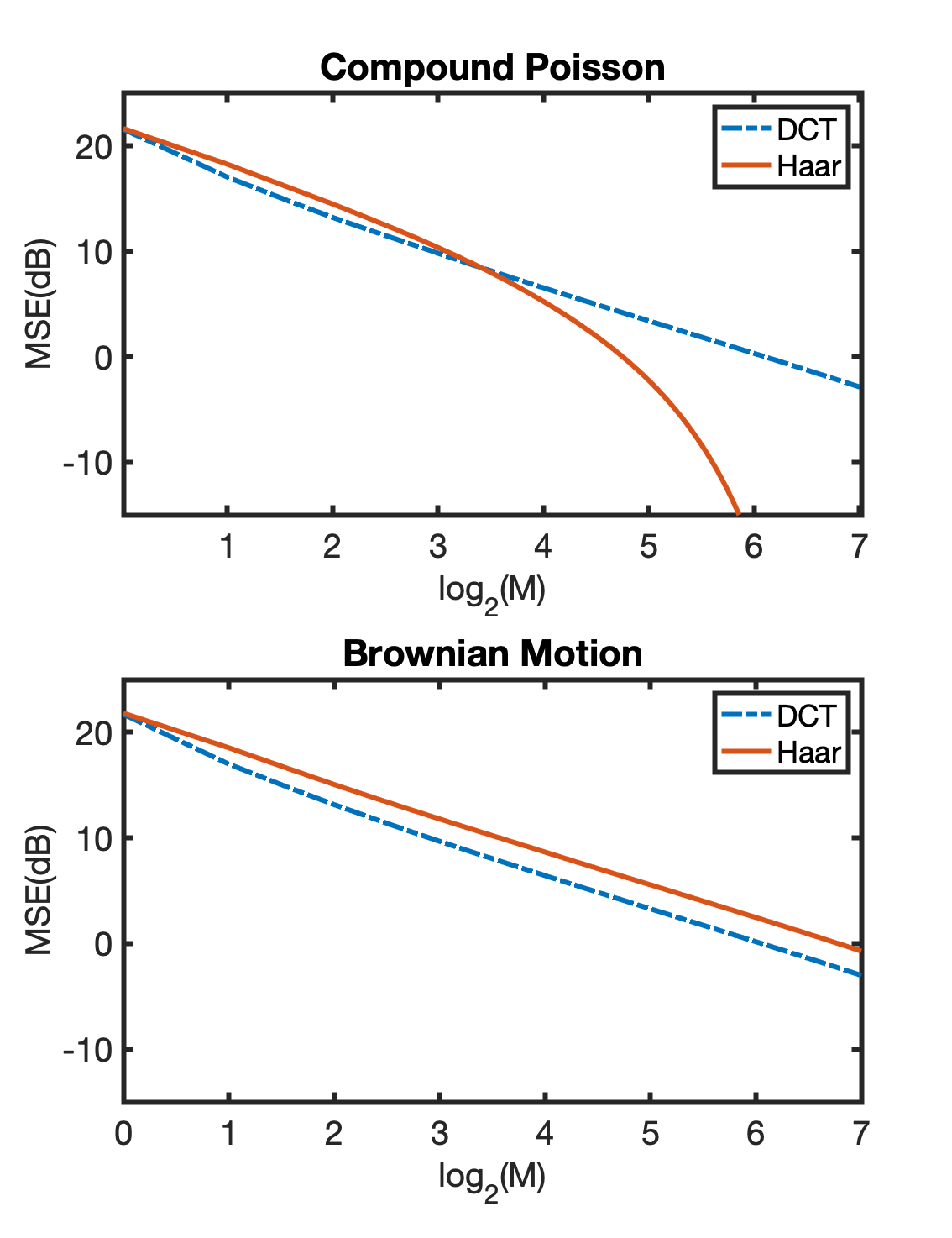}
    \caption{Comparison of DCT versus Haar wavelets to optimally represent (best $M$-term) a compound Poisson process (top) with $\lambda =10$ and Gaussian jumps with a Brownian motion (bottom). We normalize both random processes to have unit variance. }
    \label{fig:DCTvsHaar}
\end{figure}

\section{Conclusion}

 The theoretical and empirical findings of this paper are reminiscent to the so-called ``Mallat's heuristic"~\cite{donoho1993unconditional}, which states that
    
    \emph{``Wavelets are the best bases for representing objects composed of singularities, when there may be an arbitrary number of singularities, which may be located in all possible spatial position."}
        
    \noindent and which remarkably describes the compound Poisson model.

To do so, we provided a theoretical analysis to characterize the compressibility of compound Poisson processes. To that  end, we  introduced a simple approximation greedy scheme
performed over the Haar wavelet basis.  We then provided  comparable  lower and upper-bounds for the mean-squared approximation error. This enabled us to deduce the sub-exponential super-polynomial asymptotic  behavior for the error. Future research direction is to investigate the compressibility of compound Poisson processes in other dictionaries ({\it e.g.} DCT or other wavelet families), to investigate the effect of the Poisson parameter $\lambda$ in this analysis, specifically when $\lambda \rightarrow +\infty$,  and finally, to theoretically compare the best and greedy approximation schemes. 

\section*{Acknowledgment}
        
        The authors are extremely grateful to Prof. Michael Unser, who strongly inspired this work on its early stage. They also warmly thank Laur\`ene Donati for her kind help during the writing process.
    
\appendices

\section{Proof of Lemma \ref{Lemma:Poisson}}\label{App:LemmaPoisson}
\begin{proof}
We first remark that the inequality $\Delta \leq (b-a)/N$ is obviously true when $N=0$, since $\Delta = (b-a)$ in this case. As for $N\geq 1$, we have by  definition of $\Delta$ that $  \Delta \leq  x_i - x_{i-1}$, for all $i=1,\ldots,N$, with the convention that $x_0=a$. By summing up these equality for for all values of $i$, we obtain  that
\begin{equation}
N\Delta \leq x_N - x_0  \leq b-a. 
\end{equation}
This yields that  $\Delta  \leq (b-a)/N$. 

For the second part, we define the random  vector $\mathbf{d}= (d_1,\ldots,d_n)\in[0,1]^n$ as
\begin{equation}\label{Eq:di}
d_i= \frac{x_i-x_{i-1}}{b-a}, \quad i=1,2,\ldots,n,
\end{equation}
By rewriting \eqref{Eq:di} in the vectorial form, we obtain that
\begin{equation}\label{Eq:ChangeVar}
\mathbf{d}=\mathbf{H} \mathbf{x} - \frac{a}{b-a}\mathbf{e}_1,
\end{equation}
where  $\mathbf{e}_1=(1,0,\ldots,0)\in\mathbb{R}^n$, $\mathbf{x}=(x_1,\ldots,x_n)$ and $\mathbf{H}\in\mathbb{R}^{n\times n}$ is the lower-bidiagonal matrix  
\begin{equation}
\mathbf{H}=\frac{1}{b-a}\begin{pmatrix}1 &  &   & &   \\ -1 & 1 &  &  &   \\  & -1 & 1 &  &  \\  &  & \ddots & \ddots &  \\   & &  & -1 & 1\end{pmatrix}.
\end{equation}
Now, due to  \eqref{eq:densityt} and  the  change of variables \eqref{Eq:ChangeVar}, the PDF of $\mathbf{d}$ is
\begin{equation}\label{Eq:Pdfd}
p_{\mathbf{d}} (\boldsymbol{v}|N=n) =n!  \mathbbm{1}_{\boldsymbol{v}\in [0,1]^n,\|\boldsymbol{v}\|_1 \leq 1},
\end{equation}
where $\|{\bf v}\|_1 = |v_1|+\cdots+|v_n|=v_1+\cdots+v_n$ for ${\bf v} \in [0,1]^n$. In addition, from the definition of $\Delta$, the probability of  $\{\Delta\geq x\}$ for any $x\in[0,(b-a)/n]$ can be computed as 
\begin{align}
\mathbb{P}(\Delta \geq x|N=n) &= \mathbb{P}\left(\cap_{i=1}^n \{d_i \geq x/(b-a)\}|N=n\right) \nonumber
\\& =   \int_{[\frac{x}{b-a},1]^n} n!\mathbbm{1}_{\|\boldsymbol{v}\|_1 \leq 1} \mathrm{d}\boldsymbol{v} \nonumber
\\& = n!\int_{[0,1-\frac{x}{b-a}]^n} \mathbbm{1}_{\|\boldsymbol{u}\|_1 \leq 1-n\frac{x}{b-a}}  \mathrm{d}\boldsymbol{u}, \label{Eq:Integral}
\end{align} 
 where the latter is obtained via the change of variable $u_i=v_i-\frac{x}{b-a}$ for $i=1,\ldots,n$. We remark that if $u_i\geq 0$ for $i=1,\ldots,n$ and $\|\boldsymbol{u}\|_1 \leq 1-nx/(b-a)$, then we would have  $u_i \leq 1-nx/(b-a)\leq 1- x/(b-a)$ for any $i=1,\ldots,n$. In other words, the upper-limit of the integral in \eqref{Eq:Integral} is redundant and can be replaced with $+\infty$.  Doing so, we obtain that 
 \begin{align*}
  \mathbb{P}(\Delta \geq x|N=n) & = n!\int_{[0,+\infty)^n} \mathbbm{1}_{\|\boldsymbol{u}\|_1 \leq 1-n\frac{x}{b-a}}  \mathrm{d}\boldsymbol{u}\\& \stackrel{(i)}{=} \frac{n!}{2^n} \int_{\mathbb{R}^n} \mathbbm{1}_{\|\boldsymbol{u}\|_1 \leq 1-n\frac{x}{b-a}}  \mathrm{d}\boldsymbol{u} \\ & = \frac{n!}{2^n} \mathrm{Leb}\left(\left\{\|\boldsymbol{u}\|_1 \leq \left(1-n\frac{x}{b-a}\right)\right\}\right)
  \\& = \frac{n!}{2^n} \left(1-n\frac{x}{b-a}\right)^n \mathrm{Leb}\left(\{\|\boldsymbol{u}\|_1 \leq 1\}\right),
  \end{align*}
where   (i) is due to the symmetry of the integrand with respect to the sign of $\mathbf{u}$ and where $\mathrm{Leb}$ denotes the Lebesgue measure. Finally, we use a  known result stating that the volume of the $\ell_1$ unit ball in $\mathbb{R}^n$ is $2^n/n!$ \cite{wang2005volumes}. This yields to 
\begin{align*}
  \mathbb{P}(\Delta \geq x|N=n) &= \frac{n!}{2^n} \left(1-n\frac{x}{b-a}\right)^n \frac{2^n}{n!} \\&= \left(1-n\frac{x}{b-a}\right)^n.
\end{align*}
\end{proof}

\section{Proof of Proposition \ref{prop:HaarTield}}\label{App:HaarTield}
\begin{proof}
A simple computation reveals that $- \mathrm{D} \tilde{\psi}_{j,k} = \psi_{j,k}$. Hence, using the known identity $\mathrm{D}^*=-\mathrm{D}$ and \eqref{Eq:Dsw}, we have that
\begin{equation}
\langle s , \psi_{j,k} \rangle 
=
\langle s , - \mathrm{D} \tilde{\psi}_{j,k} \rangle
=
\langle \mathrm{D} s , \tilde{\psi}_{j,k} \rangle 
= 
\langle w , \tilde{\psi}_{j,k} \rangle.
\end{equation}
With a similar idea, we remark that $\mathrm{D} \tilde{\phi} =\delta -  \phi$. Combining with    $ \langle s , \delta \rangle = s(0) = 0$, we have that
\begin{equation}
    \langle s , \phi \rangle
    =
    \langle s,  \delta - \mathrm{D} \tilde{\phi} \rangle
    =
    s(0) + \langle \mathrm{D} s , \tilde{\phi} \rangle
    = 
    \langle w , \tilde{\phi} \rangle.
\end{equation}
\end{proof}
\section{Proof of Proposition \ref{Prop:LinError}}\label{App:linerr}
\begin{proof}
One  observes from  Definition \ref{Def:ApproximationSchemes} that 
\begin{equation*}\label{Eq:LinearApprox}
\mathrm{P}_M^{\rm lin}(s)= \langle s, \phi\rangle \phi + \sum_{j=0}^{J-1} \sum_{k=0}^{2^j-1} \langle s,\psi_{j,k}\rangle \psi_{j,k}+ \sum_{k=0}^{m-1} \langle s,\psi_{J,k}\rangle \psi_{J,k}.
\end{equation*}
This together with  \eqref{Eq:WaveletExpansionPoisson} yields that  
\begin{equation*}
s-\mathrm{P}_M^{\rm lin}(s) = \sum_{j\geq J+1}\sum_{k=0}^{2^j-1} \langle s,\psi_{j,k}\rangle \psi_{j,k}+\sum_{k=m}^{2^J-1} \langle s,\psi_{J,k}\rangle \psi_{J,k}.
\end{equation*}
Haar wavelets that are supported in $[0,1]$, form an orthonormal basis for $L_2([0,1])$. Using this, we   express the approximation error based on the wavelet coefficients, as 
\begin{equation*}
\|s-\mathrm{P}_M^{\rm lin}(s)\|_{L_2}^2 = \sum_{j\geq J+1}\sum_{k=0}^{2^j-1} |\langle s,\psi_{j,k}\rangle|^2+\sum_{k=m}^{2^J-1} |\langle s,\psi_{J,k}\rangle|^2.
\end{equation*}
By taking expectation over both sides and by using Proposition \ref{prop:HaarTield}, we have  that 
\begin{align*}
\mathbb{E}[\|s-\mathrm{P}_M^{\rm lin}(s)\|_{L_2}^2] & = \sum_{j\geq J+1} \sum_{ k=0}^{2^j-1} \mathbb{E}[ |\langle s,\psi_{j,k} \rangle |^2] \\& \quad +\sum_{ k=m}^{2^J-1} \mathbb{E}[ |\langle s,\psi_{J,k} \rangle |^2]\nonumber\\& = \  \sum_{j\geq J+1} \sum_{ k=0}^{2^j-1} \mathbb{E}[|\langle w, \tilde{\psi}_{j,k}\rangle |^2] \\ & \quad + \sum_{ k=m}^{2^J-1} \mathbb{E}[|\langle w, \tilde{\psi}_{J,k}\rangle |^2]\\&= \sum_{j\geq J+1} \sum_{ k=0}^{2^j-1} \sigma_0^2  \|\tilde{\psi}_{j,k}\|_{L_2}^2 \\& \quad  + \sum_{ k=m}^{2^J-1} \sigma_0^2  \|\tilde{\psi}_{J,k}\|_{L_2}^2  .
\end{align*}
Finally, we replace  $\|\tilde{\psi}_{j,k}\|_{L_2}^2 = \frac{2^{-2j}}{12}$ (obtained via a direct computation; see Figure \ref{fig:HaarTield} for visualisation) for all $j\geq 0$ and $k=0,\ldots,2^j-1$ in the summation above to deduce that 
\begin{align*}
\mathbb{E}\left[\|s-\mathrm{P}_M^{\rm lin}(s)\|_{L_2}^2\right] & = \frac{\sigma_0^2}{12}\left( \sum_{j\geq J+1} \sum_{k=0}^{2^j-1}2^{-2j} + \sum_{k=m}^{2^J-1}2^{-2J}  \right)\\&= \frac{\sigma_0^2}{12M}\left(\frac{1}{2^J}+\frac{2^J-m}{2^{2J}}\right)
\\& = \frac{\sigma_0^2}{12M}\frac{1}{2^J}\left(2-\frac{m}{2^J}\right).
\end{align*}
   \begin{figure}[t]
    \centering
    \includegraphics[ width=0.5\textwidth]{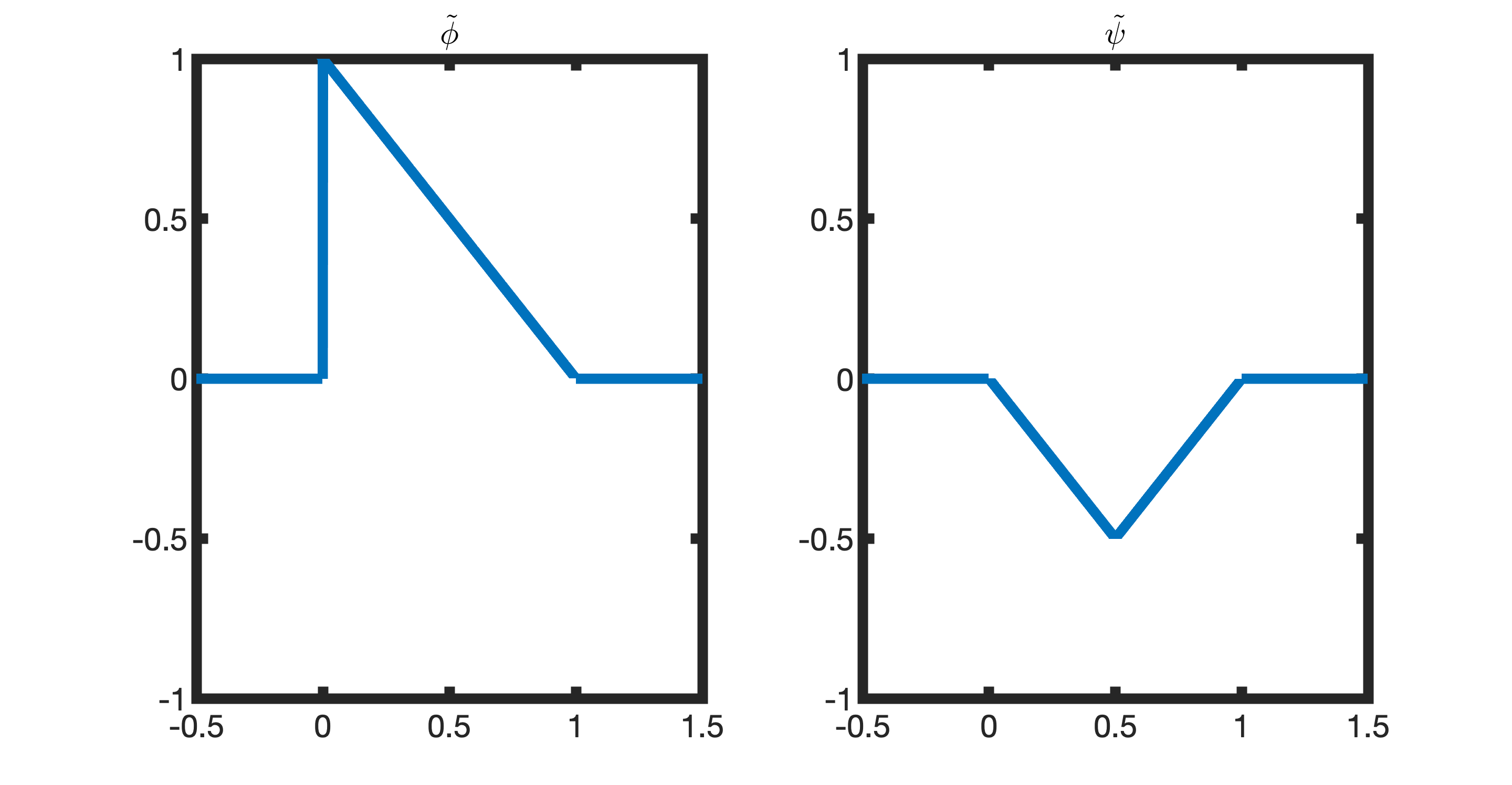}
    \caption{Auxiliary functions $\tilde{\phi}$ and $\tilde{\psi}_{j,k}$ for $j=0,1,2$ and $k=0,\ldots,2^j-1$.}
    \label{fig:HaarTield}
\end{figure}
\end{proof}

\section{Proof of Proposition \ref{Proposition:Wavelet}}\label{App:PropositionWavelet}
\begin{proof}
{\bf Item 1)} Assume that $K_{j,k}=0$. This means that $s$ is constant over the support of $\psi_{j,k}$, taking the fixed (random) value $s_0$. By recalling that 
\begin{equation}
\int_{\mathbb{R}} \psi_{j,k}(x) {\rm d } x  = 0, 
\end{equation}
for all $j\geq 0 $ and $k=0,\ldots,2^j-1$, we deduce that   the corresponding wavelet coefficient is $\langle s, \psi_{j,k}\rangle=s_0\int_{\mathbb{R}} \psi_{j,k}(x) {\rm d } x  = 0 $. 

For the converse, we show that, condition to the event $\{K_{j,k}=K \}$ for an arbitrary (but fixed) integer $K\geq 1$, we have that 
$$ \mathbb{P}\left( \langle s, \psi_{j,k}\rangle=0| K_{j,k}=K  \right)=0.$$
Consider the jumps that are inside the support of $\psi_{j,k}$ and denote their (unordered) locations and heights by $\{\tilde{\tau}_1,\ldots,\tilde{\tau}_{K}\}$ and $\{\tilde{a}_1,\ldots,\tilde{a}_{K}\}$, respectively. Due to \eqref{Eq:InnovWavelet}, we have that 
\begin{equation}
\langle s, \psi_{j,k}\rangle = \langle w, \tilde{\psi}_{j,k} \rangle= \sum_{i=1}^{K} \tilde{a}_{i} \tilde{\psi}_{j,k}(\tilde{\tau}_{i}).
\end{equation} 
We recall that  the jump locations $\tilde{\tau}_i$ are i.i.d. with  a uniform law. Moreover, the jump heights $\tilde{a}_i$ are independent of $\tilde{\tau}_i$s and are themselves  i.i.d.   copies of a random variable that admits PDF. This implies that the random variables $Z_i=\tilde{a}_{i}\tilde{\psi}_{j,k}(\tilde{\tau}_{i})$ for $i=1,\ldots,K$ are also i.i.d. and their law  has a PDF too, which  we denote by $p_Z$. Finally, the random variable $\langle s, \psi_{j,k}\rangle$ also has PDF (that is the $K$ times convolution of $p_Z$ with itself) and thus, is nonzero with probability one (no atoms). 

{\bf Item 2)} Recall that $N$ is the total number of jumps of $s$ over $[0,1]$. Due to  \eqref{Eq:14} and the fact that the wavelets $\psi_{j,k}$ for $k=0,\ldots,2^j-1$ have disjoint support,  at each scale $j\geq 1$, at most $N$ wavelet coefficients are nonzero. On the other hand, the support of any wavelet function of scale $j$ is of size $2^{-j}$. Hence,  due to the definition of $\Delta$, the number of jumps in the support of $\psi_{j,k}$ is either one or is upper-bounded by the length of the interval divided by the minimum distance  ($=2^{-j} \Delta^{-1}$). In other words, the support of each wavelet of scale $j$  contains at most $\max(1,2^{-j} \Delta^{-1})$ jumps. 

Denote by $n_j$, the number of nonzero wavelet coefficients in the $j$th scale. Using the previous observation,  we deduce for all $j\geq 1$ that 
\begin{equation}\label{Eq:NumNonZeroOneScale}
\frac{N}{\max(1,2^{-j}\Delta^{-1})} =N \min(1,2^j \Delta)\leq n_j \leq N.
\end{equation}
 As for $j=0$ (mother and father wavelets), we deduce similar to Item 1) that   condition to  $N\geq 1$, we have  $n_0=2$.  

By defining $J_{\lim}= \lfloor\log_2(\Delta^{-1})\rfloor$, one readily  verifies that for  $j\leq J_{\lim}$, we have  $\min(1,2^j \Delta)=2^j \Delta$. By contrast, $\min(1,2^j \Delta)=1$ for  $j\geq J_{\lim}$. Using these simple observations and  by summing up lower-bounds of \eqref{Eq:NumNonZeroOneScale} for $j=1,\ldots,J$  (together with $n_0=2$), we obtain, since $\sum_{j=0}^J n_j = N_J$, that 
\begin{align*}
2+ N\Delta(2^{J+1}-2) &\leq N_J,\quad \forall J \leq J_{\lim},\\
2+ N\Delta(2^{J_{\lim}+1}-2) + (J-J_{\lim}) N &\leq N_J,\quad \forall J\geq J_{\lim}.
\end{align*}
To simplify the first lower-bound, we use the inequality $2^x\geq x$ for $x=J+1+\log_2\Delta$, which results to 
\begin{equation}\label{Ineq:JleqJlim}
2-2N\Delta + N(J+1 +\log_2 \Delta ) \leq N_J,\quad \forall J \leq J_{\lim}.
\end{equation}
As for the the second lower-bound, we use 
$$J_{\lim} \leq \log_2(\Delta^{-1}) \leq J_{\lim}+1$$
 to obtain that
\begin{equation}\label{Ineq:JgeqJlim}
2+ N\Delta(\Delta^{-1}-2) + (J+\log_2 \Delta ) N \leq N_J,\quad \forall J\geq J_{\lim}.
\end{equation}
It is  now readily to verify that the two lower-bounds in \eqref{Ineq:JleqJlim} and \eqref{Ineq:JgeqJlim} are indeed equal and hence, we have that
\begin{align}
2+N(J+1-2\Delta+\log_2 \Delta) \leq N_J,\quad \forall J\geq 0. 
\end{align}
Finally, using $\Delta N \leq 1$, we conclude that 
\begin{align}\label{Eq:Jleq}
N(J+1+\log_2 \Delta) \leq N_J,\quad \forall J\geq 0. 
\end{align}
We   follow the same principle to obtain an upper-bound for $N_J$ as well. By summing up upper-bounds of \eqref{Eq:NumNonZeroOneScale} for $j=1,\ldots,J$, together with $n_0=2$, we obtain that 
\begin{equation}\label{Eq:Jgeq}
N_J \leq 2 + NJ, \quad \forall J\geq 0.
\end{equation}
Now, by the definition of $J_M$, we know that $N_{J_M} \geq M$. Combining it with  \eqref{Eq:Jgeq}  applied to $J=J_M$ yields that 
$$M \leq N_{J_M} \leq  NJ_M+2,$$
which implies the lower-bound
\begin{equation}\label{Eq:LbFinal}
 J_M \geq \frac{M-2}{N}.    
\end{equation}
Similarly, from the definition of $J_M$, we have $N_{J_M-1} \leq M-1$. This together with  \eqref{Eq:Jleq}  applied to $J=J_M-1$ gives
\begin{align*}
 {N(J_M +\log_2 \Delta)}\leq  N_{J_M-1} \leq M-1,
\end{align*}
from which we deduce the upper-bound
\begin{align}\label{Eq:UbFinal}
J_M \leq \frac{M-1}{N}   +\log_2 \Delta^{-1}.
\end{align}
We complete the proof of \eqref{Eq:JmBound} by combining \eqref{Eq:LbFinal} and \eqref{Eq:UbFinal}, knowing that $J_M\in\mathbb{N}$.
\end{proof}

\section{Proof of Theorem \ref{Thm:Main}} \label{App:Main} 
\begin{proof}[Proof of Theorem \ref{Thm:Main}]
 Let $\sigma_0^2<+\infty$ be the variance of the process $s$. We  divide the proof and  show each side of the inequality \eqref{Eq:MSEbounds} separately.

\textbf{Upper-bound}:   First, we show that for any $n\geq 1$, we have 
\begin{align*}
 \mathbb{E}\left[\|s- \mathrm{P}_{M}^{\rm greedy}(s) \|_2^2|N=n\right]  \leq \frac{\sigma_0^2 n}{6\lambda}  2^{-  \frac{M-2}{n}  }.
\end{align*}

Let us then work conditionally to $N=n$. From Proposition \ref{Proposition:Wavelet}, we have (condition to $N=n$) that  $J_M \geq \lceil \frac{M-2}{n}\rceil$. Thus, by combining with \eqref{Eq:JmSpLinBound}, we obtain that 
$$
\|s- \mathrm{P}_{M}^{\rm greedy}(s) \|_2  \leq \|s- \mathrm{P}_{2^{ \lceil \frac{M-2}{n}\rceil }}^{\rm lin}(s) \|_2. 
$$
Taking expectation from both sides yields 
\begin{align}
&\mathbb{E}\left[\|s- \mathrm{P}_{M}^{\rm greedy}(s) \|_2^2|N=n\right] \\& \quad\leq \mathbb{E}\left[\|s- \mathrm{P}_{2^{ \lceil \frac{M-2}{n}\rceil }}^{\rm lin}(s) \|_2^2|N=n\right]  \nonumber \\ & \quad= \sum_{j\geq \lceil \frac{M-2}{n}\rceil } \sum_{k=0}^{2^j-1} \mathbb{E}[|\langle s, \psi_{j,k} \rangle|^2|N=n], \label{Eq:UppBound}
\end{align}
On the other hand, condition  to  $N=n$, we have the equality in law
$$w=\mathrm{D}s= \sum_{i=1}^n a_i \delta(\cdot-\tau_i),$$
where $\{\tau_i\}_{i=1}^n$ is  the sequence of unordered jumps of $s$ in $[0,1]$ and $\{a_i\}_{i=1}^n$ is the sequence of corresponding  heights.   Therefore, we have that   
$$\langle s,\psi_{j,k}\rangle =\langle w,\tilde{\psi}_{j,k}\rangle =   
\sum_{i=1}^n Z_i,$$
where the random variables $Z_i = a_i \tilde{\psi}_{j,k}(\tau_i)$  are i.i.d. copies of a zero-mean random variable. We recall that the law of jumps $a_i$ has zero-mean and variance $\sigma_0^2/\lambda$. Hence, the  second-order moment of $Z_i$ can be  computed as
\begin{align*}
  \mathbb{E}[Z_i^2|N=n] & =\mathbb{E}[a_i^2 \tilde{\psi}_{j,k}(\tau_i)^2|N=n] \\&\stackrel{(i)}{=}\mathbb{E}[a_i^2|N=n] \mathbb{E}[ \tilde{\psi}_{j,k}(\tau_i)^2|N=n] \\& {=} \frac{\sigma_0^2}{\lambda} \int_{\mathbb{R}} \tilde{\psi}_{j,k}(x)^2 p_{\tau_i|N=n}(x) \mathrm{d}x \\& \stackrel{(ii)}{=} \frac{\sigma_0^2}{\lambda} \int_{0}^1 \tilde{\psi}_{j,k}(x)^2 \mathrm{d}x 
\\&=\frac{\sigma_0^2}{\lambda}\|\tilde{\psi}_{j,k}\|_2^2
\\& = \frac{\sigma_0^2\times 2^{-2j}}{12\lambda},
\end{align*}
where  we used the independence of $a_i$ and $\tau_i$ in (i) and the uniform law of $\tau_i$  in (ii) and finally,  we replaced  $\|\tilde{\psi}_{j,k}\|_{L_2}^2 = \frac{2^{-2j}}{12}$ in the last equality.  Now, due to the independence of the $Z_i$, we deduce that
\begin{equation}\label{Eq:WaveletMoments}
 \mathbb{E}[\langle s,\psi_{j,k}\rangle ^2|N=n] =\sum_{i=1}^n \mathbb{E}[Z_i^2|N=n]= n \frac{\sigma_0^2\times 2^{-2j}}{12\lambda}.
\end{equation}
By substituting \eqref{Eq:WaveletMoments} in \eqref{Eq:UppBound}, we obtain that 
\begin{align*}
 \mathbb{E}\left[\|s- \mathrm{P}_{M}^{\rm greedy}(s) \|_2^2|N=n\right] & \leq    \sum_{j\geq \lceil \frac{M-2}{n}\rceil} \sum_{k=0}^{2^j-1} n \frac{\sigma_0^2 2^{-2j}}{12\lambda} \\& =   \frac{\sigma_0^2 n}{12\lambda}     \sum_{j\geq \lceil \frac{M-2}{n}\rceil } 2^{-j} \\&= \frac{\sigma_0^2 n}{6\lambda}  2^{-\lceil \frac{M-2}{n}\rceil }\\&\leq \frac{\sigma_0^2 n}{6\lambda}  2^{- \frac{M-2}{n} },
\end{align*}

 By taking the expectation, we obtain that
\begin{align}
     \mathbb{E}\left[\|s- \mathrm{P}_{M}^{\rm greedy}(s) \|_2^2\right] &\leq  \sum_{n=1}^{\infty} \frac{\sigma_0^2 n}{6\lambda}  2^{- \frac{M-2}{n} } \mathbb{P}(N=n) \nonumber \\ &= 
     \sum_{n=1}^{M} \frac{\sigma_0^2 n}{6\lambda}  2^{- \frac{M-2}{n} } \mathbb{P}(N=n)  \nonumber
     \\& \quad + \sum_{n=M+1}^{\infty} \frac{\sigma_0^2 n}{6\lambda}  2^{- \frac{M-2}{n} } \mathbb{P}(N=n)\nonumber \\ & \leq \frac{2\sigma_0^2  }{3\lambda} M \mathbb{E}[2^{-\frac{M}{N}}]\nonumber  \\ & \quad +  \frac{2\sigma_0^2 }{3\lambda} \sum_{n=M+1}^{\infty} n \mathbb{P}(N=n),
     \label{Eq:MSEexpansion} 
\end{align}
where in the last inequality, we have used $2^{-\frac{M}{n}} \leq 1$ and $2^{\frac{2}{n}}\leq 4$ for all values of $M,n\geq 1$. Now, by invoking the relation $n\mathbb{P}(N=n) =n{\rm e}^{-\lambda} \lambda^n /n! = \lambda \mathbb{P}(N=n-1)$ for any $n\geq 1$, we deduce that 
\begin{equation*}
    \sum_{n=M+1}^{\infty} n \mathbb{P}(N=n) = \sum_{n=M+1}^{\infty} \mathbb{P}(N=n-1) = \lambda \mathbb{P}(N\geq M). 
\end{equation*}
On one hand, from the Chernov bound we have
\begin{equation}\label{Eq:PoissonTail}
 \mathbb{P}(N\geq M)\leq \mathbb{E}[{\rm e}^{tN}]{\rm e}^{-tM} = \mathrm{e}^{\lambda({\rm e}^t -1)} {\rm e}^{-tM}, \quad \forall t>0,
\end{equation}
where we used 
$$\mathbb{E}[{\rm e}^{tN}] = \sum_{n=0}^{\infty} {\rm e}^{-\lambda} \frac{{\rm e}^{tn} \lambda^n}{n!} = \sum_{n=0}^{\infty} {\rm e}^{-\lambda} \frac{(\lambda {\rm e}^{t})^n }{n!}={\rm e}^{\lambda(e^t-1)}$$.
Using \eqref{Eq:PoissonTail} with $t=\mathrm{ln}2$ (such that ${\rm e}^t=2$) yields 
\begin{equation*} 
 \sum_{n=M+1}^{\infty} n \mathbb{P}(N=n) \leq \mathbb{E}[{\rm e}^{tN}]{\rm e}^{-tM} \leq \lambda \mathrm{e}^{\lambda} 2^{-M}.
\end{equation*}
Hence, 
\begin{align*}
     \mathbb{E}\left[\|s- \mathrm{P}_{M}^{\rm greedy}(s) \|_2^2\right] &\leq     \frac{2\sigma_0^2  }{3\lambda} M \mathbb{E}[2^{-\frac{M}{N}}]  +  \frac{2\sigma_0^2 }{3}\mathrm{e}^{\lambda} 2^{-M}  \\ & = M \mathbb{E}[2^{-\frac{M}{N}}]\frac{2\sigma_0^2  }{3\lambda}  \left(1 + \frac{\lambda \mathrm{e}^{\lambda} 2^{-M}}{ M \mathbb{E}[2^{-\frac{M}{N}}]} \right)\\ & \leq M \mathbb{E}[2^{-\frac{M}{N}}]\frac{2\sigma_0^2  }{3\lambda}  \left(1 + \frac{\lambda \mathrm{e}^{\lambda} 2^{-M}}{ 2^{-M} \mathbb{P}(N=1) } \right)
     \\& =  M \mathbb{E}[2^{-\frac{M}{N}}]\frac{2\sigma_0^2  }{3\lambda}  ( 1+ \mathrm{e}^{2\lambda}), 
\end{align*}
which is the announced upper-bound with the constant $C_2 =\frac{2\sigma_0^2  }{3\lambda}  ( 1+ \mathrm{e}^{2\lambda}) $.

\textbf{Lower-bound:} Similar to the upper-bound, we show that for any $n\geq 1$, we have 
the inequality 
\begin{equation}\label{Eq:ConditionalMSE2}
\mathbb{E}[\|s- \mathrm{P}_{M}^{\rm greedy}(s) \|_2^2|N=n] \geq \frac{\sigma_0^2}{48{\rm e}\lambda}n^{-1} 2^{-\frac{M-1}{n}},
\end{equation}
which immediately implies the announced lower-bound.

We treat the case $N=1$ separately. Condition to $N= 1$, both wavelet coefficients of order zero (associated to mother and father wavelets) are nonzero. Moreover, for any $j\geq 1$, there is exactly one wavelet coefficient  of scale j that is nonzero. This implies that $J_M=M-2$ and in addition, we have that
\begin{equation}\label{Eq:LinSparse}
    \mathbb{E}[\|s- \mathrm{P}_{M}^{\rm greedy}(s) \|_2^2|N=1]  =\mathbb{E}[ \|s- \mathrm{P}_{2^{ M-1}}^{\rm lin}(s) \|_2^2|N=1].
\end{equation}
Similar to the proof of  Proposition \ref{Prop:LinError} and together with \eqref{Eq:WaveletMoments}, we  deduce that 
\begin{align*}
   &\mathbb{E}[ \|s- \mathrm{P}_{2^{ M-1}}^{\rm lin}(s) \|_2^2|N=1] 
   \\&\quad = \sum_{j\geq (M-1)} \sum_{k=0}^{2^J-1} \mathbb{E}[\langle s, \psi_{j,k} \rangle^2|N=1] 
   \\&\quad = \sum_{j\geq (M-1)} 2^j \frac{\sigma_0^2\times 2^{-2j}}{12 \lambda}
   \\& \quad = \frac{\sigma_0^2}{6\lambda }2^{-(M-1)} \\& \quad \geq \frac{\sigma_0^2}{48\lambda{\rm e}}2^{-(M-1)},
\end{align*}
which together with \eqref{Eq:LinSparse} proves \eqref{Eq:ConditionalMSE2} in this case.

Consider an arbitrary integer $n\geq 2$ and let us  work conditionally to  $N=n$. From the definition of $J_M$, we almost surely have that 
\begin{align*}
\|s- \mathrm{P}_{M}^{\rm greedy}(s) \|_2  \geq \|s- \mathrm{P}_{2^{ J_M+1}}^{\rm lin}(s) \|_2.
\end{align*}
This together with the right inequality of \eqref{Eq:JmBound}  implies almost surely that 
 \begin{align*}
&\|s- \mathrm{P}_{M}^{\rm greedy}(s) \|_2\geq \|s- \mathrm{P}^{\rm lin}_{2^{  \lfloor \frac{M-1}{n}+\log_2\Delta^{-1}\rfloor+1   }}(s) \|_2. 
 \end{align*}

By defining  $\delta=(2n^2-2n+2)^{-1}>0$ (the precise value will be used later) and $\overline{J}= \lfloor \frac{M-1}{n}+ \log_2 (\delta^{-1}) \rfloor+1$, we observe  that 
\begin{align}
&\mathbb{E}\left[\|s- \mathrm{P}_{M}^{\rm greedy}(s) \|_2^2|N=n\right] \nonumber\\&\quad \geq  \mathbb{E}\left[ \|s- \mathrm{P}^{\rm lin}_{2^{  \lfloor\frac{M-1}{n}+\log_2\Delta^{-1}\rfloor+1}}(s) \|_2^2|N=n\right]
 \nonumber\\& \quad\geq \mathbb{E}\left[ \mathbbm{1}_{\Delta\geq \delta} \|s- \mathrm{P}^{\rm lin}_{2^{  \lfloor \frac{M-1}{n}+\log_2\Delta^{-1}\rfloor +1  }}(s) \|_2^2|N=n\right] 
 \nonumber\\& \quad \geq \mathbb{E}\left[ \mathbbm{1}_{\Delta\geq \delta} \|s- \mathrm{P}^{\rm lin}_{2^{  \overline{J}}}(s) \|_2^2|N=n\right] \nonumber\\&\qquad \quad= \mathbb{E}[ \mathbbm{1}_{\Delta\geq \delta} \sum_{j\geq  \overline{J} }\sum_{k=0}^{2^j-1}  \langle s,\psi_{j,k}\rangle^2 |N=n]
\nonumber \\& \qquad \quad = \sum_{j\geq  \overline{J} }\sum_{k=0}^{2^j-1} \mathbb{E}[ \mathbbm{1}_{\Delta\geq \delta}    \langle s,\psi_{j,k}\rangle^2 |N=n].\label{Eq:IneqQ}
\end{align}

Similar to the upper-bound, we consider the   jumps of $s$ in $[0,1]$ and we denote their (unordered) locations and heights  by $\tau_1,\ldots,\tau_n$ and $a_1,\ldots,a_n$, respectively. With regard to the convention $\tau_0=0$, we consider the random variable $\tilde{\Delta}= \min_{0\leq i<j<n-1} |\tau_i-\tau_j|$ and consequently, the event 
$$E= \{\tilde{\Delta} \geq \delta\} \cap \{0\leq \tau_1, \ldots ,\tau_{n-1} \leq 1/2 - \delta\}.$$ 
We observe that condition to $E\cap\{N=n\}$, we have that
$$\Delta = \min\left(\tilde{\Delta},\min_{1\leq i \leq n-1} (\tau_n-\tau_i)\right)\geq \min(\tilde{\Delta},\delta) \geq \delta. $$
This implies that condition to $N=n$, we have
\begin{equation}\label{Eq:IndicatorIneq}
\mathbbm{1}_{E} \mathbbm{1}_{[1/2,1]}(\tau_n)  \leq \mathbbm{1}_{\Delta\geq \delta}.
\end{equation}
On the other hand, 
\begin{align*}
&\mathbb{E}\left[ \mathbbm{1}_{\Delta\geq \delta}     \langle s,\psi_{j,k}\rangle^2 |N=n\right] & \\&\quad = \mathbb{E}\left[ \mathbbm{1}_{\Delta\geq \delta}    \left(\sum_{i=1}^n a_i \tilde{\psi}_{j,k}(\tau_i)\right)^2  |N=n\right] 
\\&\quad = \mathbb{E}\left[  \left(\sum_{i=1}^n a_i \mathbbm{1}_{\Delta\geq \delta}   \tilde{\psi}_{j,k}(\tau_i)\right)^2  |N=n\right]
\\&\quad \stackrel{(i)}{=} \sum_{i=1}^n \mathbb{E}\left[  \left( a_i \mathbbm{1}_{\Delta\geq \delta}   \tilde{\psi}_{j,k}(\tau_i)\right)^2  |N=n\right]
\\&\quad \stackrel{(ii)}{=} \sum_{i=1}^n \mathbb{E}[a_i^2] \mathbb{E}\left[ \mathbbm{1}_{\Delta\geq \delta}   \tilde{\psi}_{j,k}^2(\tau_i)|N=n\right],
\end{align*}
where we used  the independence (condition to $N=n$) of jumps $\tau_i$ and heights $a_i$ of $s$  in (i)   and we used the  independence of $a_i$  from $N$ and $\Delta$ as well the fact that the law of $a_i$ has  zero mean in (ii). By substituting  $\mathbb{E}[a_i^2]=\frac{\sigma_0^2}{\lambda}$ and invoking \eqref{Eq:IndicatorIneq}, we obtain
\begin{align*}
&\mathbb{E}\left[ \mathbbm{1}_{\Delta\geq \delta}    |\langle s,\psi_{j,k}\rangle|^2 |N=n\right]      \\&\quad=  n\frac{\sigma_0^2}{\lambda} \mathbb{E}[ \mathbbm{1}_{\Delta\geq \delta}   \tilde{\psi}_{j,k}^2(\tau_n)|N=n] \\&\qquad \geq n\frac{\sigma_0^2}{\lambda} \mathbb{E}[ \mathbbm{1}_{E}\mathbbm{1}_{\tau_n\in[1/2,1]}   \tilde{\psi}_{j,k}^2(\tau_n)|N=n]
\\ & \qquad \quad = n\frac{\sigma_0^2}{\lambda} \mathbb{P}[ E|N=n]\mathbb{E}[\mathbbm{1}_{x_n\in[1/2,1]}   \tilde{\psi}_{j,k}^2(x_n)|N=n],
\end{align*}
where the latter is deduced from the independence of  $E$ and $\{1/2\leq \tau_n \leq 1\}$ (condition to $N=n$). By using  Lemma \ref{Lemma:Poisson} with $a=0$ and $b=1/2-\delta$ (we remind that $\delta=(2n^2-2n+2)^{-1}$), we can compute the conditional probability of the event $E$ as
\begin{align*}
\mathbb{P}(E|N=n)& = \left( 1- (n-1) \frac{\delta}{1/2- \delta} \right)^{(n-1)} \\
& = \left( 1- (n-1) \frac{(2n^2 - 2n +2 )^{-1} }{1/2- (2n^2 - 2n +2 )^{-1} } \right)^{(n-1)}
\\& = \left( 1- n^{-1}  \right)^{(n-1)}.
\end{align*}
Now, using Lemma \ref{Lemma:Poisson} and the above computation, we have that 

\begin{align*}
&\mathbb{E}\left[ \mathbbm{1}_{\Delta\geq \delta}    |\langle s,\psi_{j,k}\rangle|^2 |N=n\right]    \\&\geq   n\frac{\sigma_0^2}{\lambda} \mathbb{P}(E |N=n)  \int_{\frac{1}{2}}^{1}  \tilde{\psi}_{j,k}^2(x) \mathrm{d}x
\\&= \frac{\sigma_0^2n}{\lambda}  (1-n^{-1})^{(n-1)}  \|\tilde{\psi}_{j,k}\mathbbm{1}_{[1/2,1]}\|_2^2 
\\& \stackrel{(i)}{=}\frac{\sigma_0^2n}{\lambda}  (1-n^{-1})^{(n-1)} \mathbbm{1}_{k\geq 2^{j-1}} \|\tilde{\psi}_{j,k}\|_2^2
\\&\stackrel{(ii)}{\geq} \frac{\sigma_0^2n}{\lambda} {\rm e}^{-1}  \mathbbm{1}_{k\geq 2^{j-1}} \|\tilde{\psi}_{j,k}\|_2^2,
\end{align*} 
where (i) simply exploits that $\tilde{\psi}_{j,k}\mathbbm{1}_{[1/2,1]}=0$ for $k\leq 2^{j-1}-1$ together with $\tilde{\psi}_{j,k}\mathbbm{1}_{[1/2,1]}=\tilde{\psi}_{j,k}$ for $k\geq 2^{j-1}$ and (ii) uses $(1-n^{-1})^{(n-1)}\geq {\rm e}^{-1}$. Going back to  \eqref{Eq:IneqQ}, we obtain  for any $n\geq 2$  that 
\begin{align*}
&\mathbb{E}[\|s- \mathrm{P}_{M}^{\rm greedy}(s) \|_2^2|N=n]   \\&\geq  \sum_{j\geq  \overline{J} }\sum_{k=0}^{2^j-1} \frac{\sigma_0^2n}{\lambda}  {\rm e}^{-1}  \|\tilde{\psi}_{j,k}\|_2^2 \mathbbm{1}_{k\geq 2^{j-1}}
\\& = \sum_{j\geq \overline{J}}  \frac{\sigma_0^2n}{\lambda} e^{-1} \|\tilde{\psi}_{j,k}\|_2^2 2^{j-1}
\\&\stackrel{(i)}{=}  \frac{\sigma_0^2n}{12{\rm e}\lambda}  2^{-\overline{J}}
\\& =  \frac{\sigma_0^2}{12{\rm e}\lambda}   n 2^{- \lfloor \frac{M-1}{n}+\log_2\delta^{-1} \rfloor-1} 
\\&\stackrel{(ii)}{\geq} \frac{\sigma_0^2}{24{\rm e}\lambda} n 2^{-\frac{M-1}{n}} \delta
\\& \stackrel{(iii)}{=} \frac{\sigma_0^2}{48{\rm e}\lambda} \frac{n}{n^2-n+1} 2^{-\frac{M-1}{n}}
\\& \stackrel{(iv)}{\geq} \frac{\sigma_0^2}{48{\rm e}\lambda}n^{-1} 2^{-\frac{M-1}{n}},
\end{align*}
where (i) uses $\|\tilde{\psi}_{j,k}\|_{L_2}^2 = 2^{-2j}/12$, (ii) simply follows from $\lfloor \frac{M-1}{n}+\log_2\delta^{-1} \rfloor \leq   \frac{M-1}{n} +\log_2 \delta^{-1}$, (iii)  uses  the value of $\delta=(2n^2+2n-2)^{-1}$, and (iv) that $\frac{n}{n^2-n+1}\geq \frac{1}{n}$, due to $n^2-n+1 \leq n^2$ for any $n\geq 1$.
 Finally, we take the overall expectation to deduce that
\begin{align}
    \mathbb{E}[\|s- \mathrm{P}_{M}^{\rm greedy}(s) \|_2^2] &  \geq \sum_{n=1}^{\infty} \frac{\sigma_0^2}{48{\rm e}\lambda}n^{-1} 2^{-\frac{M-1}{n}} \mathbb{P}(N=n) \nonumber \\ & \geq   \frac{\sigma_0^2}{48{\rm e}\lambda} \sum_{n=1}^M n^{-1} 2^{-\frac{M}{n}} \mathbb{P}(N=n) \nonumber \\ & \geq  \frac{\sigma_0^2}{48{\rm e}\lambda}M^{-1} \sum_{n=1}^M  2^{-\frac{M}{n}} \mathbb{P}(N=n). \label{Eq:MSElbExp}
\end{align}
We note that 
\begin{align}
    \sum_{n=1}^M  2^{-\frac{M}{n}} \mathbb{P}(N=n) \geq  2^{-M} \mathbb{P}(N=1) = \lambda {\rm e}^{-\lambda} 2^{-M}.
\end{align}
Moreover, we use \eqref{Eq:PoissonTail} to deduce that 
\begin{align}
    \sum_{n=M+1}^{\infty}  2^{-\frac{M}{n}} \mathbb{P}(N=n) \leq   \mathbb{P}(N\geq M+1) \leq \lambda {\rm e}^{\lambda} 2^{-M}.
\end{align}
Combining the two inequalities with \eqref{Eq:MSElbExp} yields
\begin{align*}
    &\mathbb{E}[\|s- \mathrm{P}_{M}^{\rm greedy}(s) \|_2^2] \\  & \qquad \geq \frac{\sigma_0^2}{48{\rm e}\lambda}M^{-1} \mathbb{E}[2^{-\frac{M}{N}}] \frac{\sum_{n=1}^M  2^{-\frac{M}{n}} \mathbb{P}(N=n)}{\mathbb{E}[2^{-\frac{M}{N}}]} 
    \\ & \qquad \geq \frac{\sigma_0^2}{48{\rm e}\lambda}M^{-1} \mathbb{E}[2^{-\frac{M}{N}}] \frac{\lambda {\rm e}^{-\lambda} 2^{-M}}{\lambda {\rm e}^{-\lambda} 2^{-M} + \lambda {\rm e}^{\lambda} 2^{-M}} 
    \\ &  \qquad= M^{-1} \mathbb{E}[2^{-\frac{M}{N}}] \frac{\sigma_0^2}{48{\rm e}\lambda(1+{\rm e}^{2\lambda})},
\end{align*}
which yields the desired lower-bound with the constant $C_1= \frac{\sigma_0^2}{48{\rm e}\lambda(1+{\rm e}^{2\lambda})}$.
\end{proof}

\section{Proof of Theorem \ref{Thm:SubExpSupPoly}}\label{App:SubExpSupPoly}
\begin{proof}
It is sufficient to prove that the quantity $\mathbb{E}[2^{-\frac{M}{N}}]$ has  sub-exponential and super-polynomial   asymptotic behavior.

\textbf{Super-polynomiality:} First note that there exists an integer number $N_0\in \mathbb{N}$ such that for every $n\geq N_0$, we have $\mathbb{P}(N=n) \leq 2^{-n}$. We then consider the following decomposition for any $M\geq N_0+1$
\begin{align*}
\mathbb{E}[2^{-\frac{M}{N}}]&=  \sum_{n=1}^{N_0-1} \mathbb{P}(N=n)  2^{-\frac{M}{n}} &\\ &+  \sum_{n=N_0}^{M-1} \mathbb{P}(N=n)    2^{-\frac{M}{n}} \\& +   \sum_{n=M}^\infty \mathbb{P}(N=n)  2^{-\frac{M}{n}}.
\end{align*}
We    separately show that each term of the previous decomposition decays  faster than the inverse of any polynomial as $M\rightarrow\infty$. 

For the first term,  simply due to $\mathbb{P}(N=n) \leq 1$, we have that  
$$
M^k \sum_{n=1}^{N_0-1} \mathbb{P}(N=n)  2^{-\frac{M}{n}}  \leq (N_0-1)  M^k 2^{-\frac{M}{N_0}} \longrightarrow 0, 
$$
as $M\rightarrow\infty$. Regarding the second term, we use the bound $\mathbb{P}(N=n) \leq 2^{-n}$   for $n \geq N_0$ to deduce that 
$$
\forall n\geq N_0: \mathbb{P}(N=n)  2^{-\frac{M}{n}} \leq 2^{-n-\frac{M}{n}} \leq 2^{-2\sqrt{M}},
$$
where in the last inequality, we have used $x+y \geq 2\sqrt{xy}$ with $x=n$ and $y= M/n$. Hence,
$$
M^k \sum_{n=N_0}^{M-1} \mathbb{P}(N=n)  2^{-\frac{M}{n}} \leq    M^{k+1}   2^{-2\sqrt{M}} \rightarrow 0,
$$
as $ M\rightarrow\infty$. Finally for the last term, we use  \eqref{Eq:PoissonTail} with  $t=1$ to obtain that 
\begin{align*}
M^k \sum_{n=M}^\infty \mathbb{P}(N=n)  2^{-\frac{M}{n}} &\leq  M^k \sum_{n=M-1}^\infty \mathbb{P}(N=n) \\&= M^k \mathbb{P}(N\geq M) \\& \leq M^k {\rm e}^{\lambda({\rm e}-1)}{\rm e}^{ - M} \rightarrow 0, 
\end{align*}
 as $ M\rightarrow \infty$.

\textbf{Sub-exponentiality:} To show the sub-exponential behavior, we fix $\alpha>0$ and for all $n_0\geq 2$, we note that 
\begin{align*}
{\rm e}^{\alpha M}\mathbb{E}[2^{-\frac{M}{N}}]&\geq 2^{\log_2 ({\rm e}) \alpha M}    \mathbb{P}(N=n_0)    2^{-\frac{M}{n_0}} \\& =    \mathbb{P}(N=n_0)   2^{(\alpha \log_2 ({\rm e})-\frac{1}{n_0}) M}.
\end{align*}
 Now by fixing $n_0$ to be a sufficiently large integer so that $\log_2 {\rm e} \alpha-\frac{1}{n_0} >0$, we deduce that  the right hand side explodes.
\end{proof}

 \bibliographystyle{IEEEtran.bst}
\bibliography{ref.bib}

\end{document}